\documentclass[12pt]{article}
\usepackage{amsmath}
\usepackage{graphicx}
\usepackage{enumerate}
\usepackage{natbib}
\usepackage{url} 
\usepackage{comment}
\usepackage{bbm}

\newcommand{\blind}{0}

\usepackage{bm}
\usepackage{amssymb}
\usepackage{amsthm}
\usepackage{hyperref}
\usepackage{color}

\newcommand{\Mean}{{\mathbb{E}}}
\newcommand{\Var}{{\mbox{Var}}}
\newcommand{\Cov}{{\mbox{Cov}}}
\newcommand{\Corr}{{\mbox{corr}}}

\newcommand{\bas}{\begin{eqnarray*}}
\newcommand{\eas}{\end{eqnarray*}}

\newtheorem{thm}{Theorem}

\newtheorem{prop}{Proposition}
\newtheorem{asmp}{Assumption}
\newtheorem{cor}{Corollary}

\theoremstyle{definition}
\newtheorem{example}{Example}

\usepackage{algorithm}
\usepackage{algorithmicx}
\usepackage{algpseudocode}

\usepackage{xr}
\begin{document}

\def\spacingset#1{\renewcommand{\baselinestretch}%
{#1}\small\normalsize} \spacingset{1}


\if0\blind
{
  \title{\bf   Causal Inference in Biomedical Imaging via Functional Linear Structural Equation Models  }
    \author{Ting Li
    \hspace{.2cm}\\
    School of  Statistics and Data Science,  \\  Shanghai University of Finance and Economics, Shanghai, China\\
   Ethan Fan \\
    Department of Biostatistics, Duke University, Durham, USA \\
   Tengfei Li and    Hongtu Zhu\thanks{ Address for correspondence: Hongtu Zhu, Ph.D., E-mail: htzhu@email.unc.edu. 
  } \\
    Departments of Radiology, Computer Science, Genetics, and Biostatistics, \\ University of North Carolina at Chapel Hill, Chapel Hill, USA}
    \date{}
  \maketitle
} \fi

\if1\blind
{
  \bigskip
  \bigskip
  \bigskip
  \begin{center}
    {\LARGE \bf Causal Inference in Biomedical Imaging via Functional Linear Structural Equation Models }
\end{center}
  \medskip
} \fi

\bigskip
\begin{abstract} 
Understanding the causal effects of organ-specific features from medical imaging on clinical outcomes is essential for biomedical research and patient care. We propose a novel Functional Linear Structural Equation Model (FLSEM) to capture the relationships among clinical outcomes, functional imaging exposures, and scalar covariates like genetics, sex, and age. 
Traditional methods struggle with the infinite-dimensional nature of exposures and complex covariates. Our FLSEM overcomes these challenges by establishing identifiable conditions using scalar instrumental variables. We develop the Functional Group Support Detection and Root Finding (FGSDAR) algorithm for efficient variable selection, supported by rigorous theoretical guarantees, including selection consistency and accurate parameter estimation. 
We further propose a test statistic to test the nullity of the functional coefficient, establishing its null limit distribution. 
Our approach is validated through extensive simulations and applied to UK Biobank data, demonstrating robust performance in detecting causal relationships from medical imaging.
\end{abstract}

\noindent%

\spacingset{1.9} 
\section{Introduction}

 This paper proposes the Causal-Genetic-Imaging-Clinical (CGIC) framework to uncover causal pathways linking genetic factors, organ-level imaging markers, and clinical outcomes for complex phenotypes such as Alzheimer’s Disease (AD), while accounting for environmental influences and unobserved confounders \citep{zhu2023statistical, le2019mapping, taschler2022causal}. Imaging modalities like functional Magnetic Resonance Imaging (fMRI) play a central role in capturing anatomical and functional properties of organs such as the brain \citep{zhou2021review}. The CGIC framework extends Jack’s influential hypothetical model of AD progression \citep{jack2010hypothetical, jack2013tracking}, offering a more comprehensive view of disease etiology (Figure \ref{fig:illustration}(a)). Large-scale studies such as the Alzheimer's Disease Neuroimaging Initiative (ADNI) provide critical multimodal data for investigating the interplay among genetics, imaging, and clinical outcomes.

A central challenge in establishing CGIC pathways is the presence of unmeasured or poorly controlled confounders \citep{vandenbroucke2004observational}. Mendelian Randomization (MR) offers a powerful approach to infer causal effects by leveraging genetic variants as instrumental variables (IVs), under the assumption that these variants affect the outcome only through the exposure of interest \citep{lawlor2008mendelian, burgess2017review, sanderson2022mendelian}. However, applying MR in high-dimensional genetic settings is complex. Genetic variables can play diverse roles—some serve as valid instruments, others act as confounders, precision covariates, or irrelevant noise.
Moreover, the number of genetic variants often far exceeds the sample size, posing significant statistical and computational challenges.


Aligned with the principles of MR, we leverage genetic variants as potential IVs to investigate causal relationships between functional exposures and clinical outcomes. However, unique methodological challenges arise in this context due to the infinite-dimensional nature of functional exposures (e.g., brain imaging data), the high-dimensional and mixed roles of genetic variants, and the confounding influence of environmental and lifestyle factors on both exposures and outcomes. These complexities call for advanced statistical frameworks capable of rigorously modeling the CGIC pathway.

\begin{figure}
	\centering
\includegraphics[width=0.9\textwidth]
 {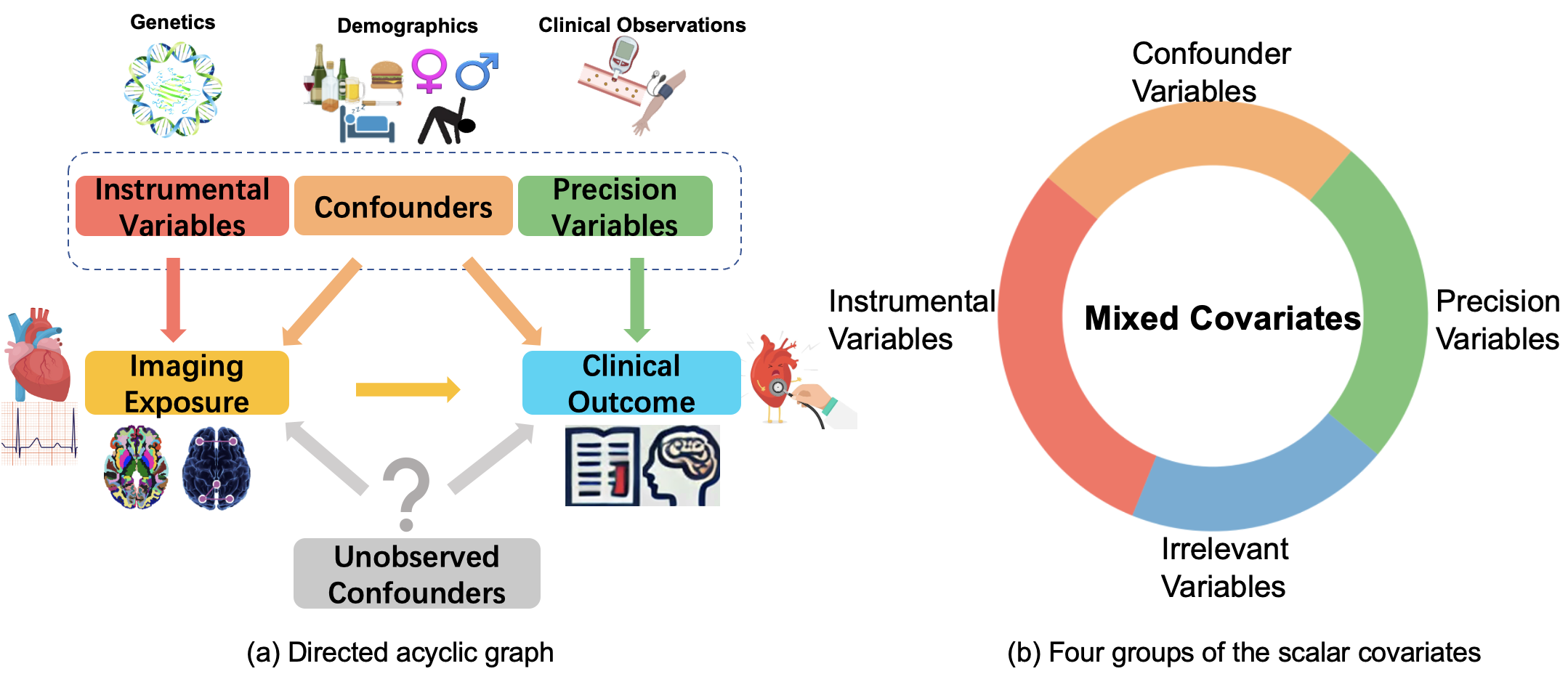}
	\caption{\small (a) Directed acyclic graph showing the causal genetic-imaging-clinical (CGIC) pathway that links from genetic factors to organ 
imaging measures  to clinical outcomes confounded with environmental factors (e.g., lifestyle factors)  and  possible unobserved confounders. (b) Four groups of  mixed  covariates.
}
	\label{fig:illustration}
	\vspace{-0.1 in}
\end{figure}

%

To enhance the robustness and interpretability of causal inference in this setting, we propose the Functional Linear Structural Equation Model (FLSEM). This model consists of two interconnected components:
(i) a structural equation linking a scalar outcome to an endogenous functional exposure and scalar covariates; and
(ii) an exposure model describing the functional predictor in terms of scalar instrumental variables.

Our main contributions are as follows:

 (C1)  We propose a three-step procedure for selecting the desired instruments and important control variables using $L_0$ penalty.
First, we identify the relevant variables using
the functional on scalar model under the reproducing kernel Hilbert space (RKHS) framework using the $L_0$ penalty.
 Second, we replace the treatment variable with its post-adaptive predicted value and select useful controls. Third, we use the selected controls and the predicted treatment variable to obtain the treatment effect estimator following the partial functional linear model. 

(C2)  Different from previous studies that use the functional instrumental variables,
we establish the identifiable conditions for the FLSEM, where both the instrumental covariates and useful controls are scalar.
To the best of our knowledge, this is the first identifiable condition for the functional exposure model with a mixed set of scalar instruments.

(C3)  We propose the Functional Group Support Detection and Root Finding (FGSDAR) algorithm to select relevant instrumental covariates for the functional exposure.

(C4)  Theoretically, we establish the selection consistency of the true instruments and control,  characterize the estimation error of both functional and
scalar estimates, and   propose a test statistic to test the nullity of the functional coefficient.  

(C5)  We conduct extensive simulation studies and analyze a real dataset from UK Biobank to demonstrate the finite sample performance of the proposed methods.


\subsection{Related Work}

Endogeneity in functional data has attracted growing attention, yet most econometric solutions still hinge on Tikhonov regularization for ill-posed inverse problems. A seminal example is \cite{florens2015instrumental}, which establishes IV estimation for functional linear models, convergence rates, and a notion of “instrumental strength.” Subsequent work includes spline‐based two-step estimators with measurement error \citep{tekwe2019instrumental}, comparable two-step procedures with nullity tests \citep{jadhav2022function}, and fully Bayesian error correction \citep{zoh2022fully}. Shifting to scalar instruments, \cite{babii2022high} extend Florens et al.’s framework. Additional advances feature a GMM estimator for serially correlated predictors \citep{chen2022functional}, an exogeneity test with asymptotic guarantees \citep{dorn2022testing}, FPCA-based IV methods \citep{seong2022functional}, and Bayesian mediation models for spatial imaging data that exploit latent-confounder structure \citep{xu2023bayesian,xu2024bayesian}. Yet none of these methods tackles the high-dimensional, mixed sets of scalar instruments and controls typical of modern genetics and neuroimaging studies. Our work closes this gap.

Our work advances function-on-scalar regression with high-dimensional scalar covariates—a setting where most methods first expand each functional coefficient in a fixed basis and then impose group penalties to achieve sparsity. Representative examples include group MCP \citep{chen2016variable} and group LASSO \citep{barber2017function}. Robust variants handle outliers via exponential-squared loss plus group SCAD \citep{cai2022robust}, or by combining least-squares loss with an adaptive LASSO \citep{fan2017high}. Other directions promote simultaneous smoothness and selection through RKHS‐based penalties \citep{parodi2018simultaneous} or a smoothing elastic net \citep{mirshani2021adaptive}. Related matrix-response work employs the trace norm for selection \citep{kong2019l2rm}. 
Departing from these LASSO-type frameworks, we adopt an  $L_0-$penalty in the exposure model, which enforces exact sparsity without the shrinkage bias inherent to convex penalties. More importantly, whereas prior studies focus on predictive function-on-scalar models, our goal is causal: we estimate the effect of an endogenous functional exposure on a scalar outcome. To do so, we embed the $L_0-$based selection within a functional-on-scalar exposure model that explicitly corrects for endogeneity, thereby extending the existing literature to high-dimensional causal inference with functional data.

 The remainder of the paper is organized as follows.  
Section \ref{sec:Identification} introduces the FLSEM and its identifiability conditions.  
Section \ref{sec:estimation} outlines the estimation procedure and proposes a test statistic to test the nullity of the functional coefficient.  
Section \ref{sec:theoretical} establishes the consistency and convergence rates of our estimators and the null-limit theorem.  
Section \ref{sec:simulation} assesses the finite-sample performance via simulations.  
Section \ref{sec:realData} applies the proposed methods to UK Biobank data.  
Proofs and additional results are included in the supplementary materials.


	\vspace{-0.1 in}
\section{Identifiability Conditions for FLSEM}
\label{sec:Identification}

 In this section, we begin by discussing the identifiability of general exposure models, laying the groundwork for examining the specific identifiability conditions of our proposed FLSEM. We consider the dataset $\{X_i, \bm{Z}_i, Y_i\}_{i=1}^n$, where for each sample $i$, $X_i=(X_{i\ell}) \in \mathbb{R}^p$ represents the scalar covariates, $\bm{Z}_i \in L^2(\mathcal{T})$ denotes the functional exposure, and $Y_i$ is the corresponding response.  
The outcome and exposure processes are, respectively, given by 
\begin{equation}
  \vspace*{-0.1 in}
Y_i  = \sum_{l=1}^{p} X_{i\ell}\beta_\ell+ \int_{\mathcal T} \bm{Z}_i(t)\bm{B}(t) dt + \epsilon_{i}~~~~\mbox{and}~~~~    
\bm{Z}_i(t) = \sum_{l=1}^{p} f_\ell (X_{i\ell}, t) +  \bm{E}_{i}(t),  \label{FLSEM} 
\end{equation}   
where 
$\beta_\ell$ captures the effect of the $\ell$-th scalar covariate $X_{i\ell}$ on the response $Y_i$, 
$\bm{B}(t)$ quantifies the impact of the functional exposure $\bm{Z}_i$ on $Y_i$, and
$f_\ell(X_{i\ell}, t)$ is a nonparametric function describing how the $\ell$-th covariate influences the functional exposure $\bm{Z}_i$. Moreover,  
$\epsilon_i$ represents the random error in the outcome model and 
$\bm{E}_i(t)$, a random process with mean zero, accounts for additional variability in the functional exposure.

A critical aspect of these models in (\ref{FLSEM}) is the potential correlation between $\bm{E}_i(t)$ and $\epsilon_i$,  introducing endogeneity. To address this, we set the following assumptions: (i) 
$E(\epsilon_i | X_i) = 0$ and $E\{\bm{E}_i(t) | X_i\} = 0$, ensuring that the error terms are conditionally zero mean independent of the covariates.
(ii) $E(\epsilon_i | \bm{Z}_i) \neq 0$ and $E\{\epsilon_i \bm{E}_i(t)\} \neq 0$, reflecting the influence of unobserved confounders. 
These model specifications and assumptions form the basis for our subsequent analysis, where we focus on the identifiability challenges posed by endogeneity and develop methods to address them within the FLSEM framework.



 Following   \cite{yu2022mapping}, we organize all covariates into four distinct groups in  $\mathcal{A} = \{1, \ldots, p\}$ (Figure \ref{fig:illustration} (b)).  Let $\mathcal{C} = \{\ell \in \mathcal{A} \mid \beta_\ell \neq 0 \text{ and } f_\ell \neq 0\}$ denote the indices of confounders, which are associated with both the outcome and the exposure. The set $\mathcal{P} = \{\ell \in \mathcal{A} \mid \beta_\ell \neq 0 \text{ and } f_\ell = 0\}$ includes precision variables that affect the outcome but not the exposure. The set $\mathcal{I} = \{\ell \in \mathcal{A} \mid \beta_\ell = 0 \text{ and } f_\ell \neq 0\}$ represents IVs that influence the exposure but not the outcome. Lastly, $\mathcal{S} = \{\ell \in \mathcal{A} \mid \beta_\ell = 0 \text{ and } f_\ell = 0\}$ corresponds to irrelevant variables with no association to either. According to   \cite{guo2018confidence}, $\mathcal{I}$ consists of valid instruments, whereas $\mathcal{C}$ contains invalid instruments due to their dual association with both outcome and exposure.

In our outcome model, $\bm{Z}_i(t)$ may be correlated with the error term $\epsilon_i$, violating the exogeneity assumption in classical partial functional linear regression models \citep{kong2016partially, yu2022mapping}. Here, the coefficient $\bm{B}(t)$ captures the causal effect of changes in $\bm{Z}_i(t)$ on $Y_i$, rather than mere association. Standard approaches \citep{kong2016partially, li2020inference} may yield biased estimates by ignoring potential endogeneity in $\bm{Z}_i(t)$. To address this, we incorporate the exposure model and IVs derived from $X_i$, establishing identifiability of $\bm{B}(t)$ and enabling valid causal estimation. This strategy clarifies the $\bm{Z}_i(t)$–$Y$ relationship and strengthens the validity of our causal inferences.

 \subsection{Identifiability Condition for General $f$}

The identifiability of parameters with scalar exposures has been well studied. For example, \cite{kang2016instrumental} propose the “50\% rule,” requiring a majority of instruments to be valid, while \cite{guo2018confidence} relax this to a “plurality rule” for broader applicability. Building on these insights, we extend identifiability conditions to the more complex setting of functional exposures. The infinite-dimensional nature of functional treatments poses unique challenges, particularly when instruments may be invalid. To our knowledge, this is the first study to address identifiability under such conditions.

 Following model \eqref{FLSEM}, let $f(X_{i}, t) = \sum_{\ell =1}^{p} f_\ell (X_{i\ell}, t)$,
  we  have the  moment condition that
\begin{eqnarray}
	\label{eq:reduced form}
g(X_i) =  E(Y_i | X_i) = X_i^\top \beta +  \int_\mathcal{T} f(X_{i}, t) \bm{B}(t) dt,
\end{eqnarray}
 leading to 
 $M_X g(X_i) =  \int_\mathcal{T} M_X  f(X_{i}, t) \bm{B}(t) dt $, where $M_X$ represents the orthogonal space  of the linear space spanned by $X$.
Let ${\mathcal{K} } \bm{B}(t) = \int_\mathcal{T} E\{ M_X  f(X_{i}, s) M_X  f(X_{i}, t) \} \bm{B}(s) ds  $.
 Using the inverse of linear integral equations \citep{hsing2015theoretical},
 we present a straightforward and intuitive condition for identifying $\bm{B}(s)$.
 
 \begin{prop}
 	\label{prop:invalid IV identification}
 	Consider the reduced form \eqref{eq:reduced form}, if the operator ${\mathcal{K} }$ is injective, and the null space of the operator  ${\mathcal{K} }$ only contains 0 such that $\mathcal{N} ( \mathcal{K} )=\{0\}$, then $\bm{B}(t)$ is identifiable.
 \end{prop}

Proposition \ref{prop:invalid IV identification} follows directly from Theorem 3.5.1 of  \cite{hsing2015theoretical}.
Specifically, if there exists $\bm{B}_1(s)$ and $\bm{B}_2(s)$ satisfying
${\mathcal{K}} \bm{B}_1(s) = {\mathcal{K}} \bm{B}_2(s)$, or equivalently ${\mathcal{K}} ( \bm{B}_1(s) -  \bm{B}_2(s) )=0$, it directly implies that $\bm{B}_1(s) =  \bm{B}_2(s)$.
 This proposition extends the relevance rank condition \citep{liang2022selecting} commonly employed in linear IV models with scalar exposures to functional exposure.
 Notably, it differs from prior studies such as \cite{florens2015instrumental} and \cite{babii2022high}, where IVs are either provided or constructed based on an injective operator derived from the logistic cumulative distribution function. Instead, we consider infinite-dimensional IVs  $f(X, t)$.
Furthermore, Proposition \ref{prop:invalid IV identification} considers the  linear effect of $X$ on $Y$, and nonlinear effect of $X$ on $\bm{Z}$,
which separates the effects of the instruments and the confounders. 

 For the integral operator $\mathcal{K}$, the kernel $K(s, t) =  E\{   f(X, s)  f(X, t) \}$ plays a fundamental role.
Mercer's theorem provides the spectral decomposition 
$K(s, t)= \sum_{k=1}^\infty \lambda_k \varphi_k(s) \varphi_k(t)$,
where $\{ \varphi_k(\cdot), \lambda_k \geq 0  \}_{k \geq 1}$ forms an orthonormal basis in $L^2( \mathcal{T} )$. 
Any $ \bm{B}(t) \in L^2( \mathcal{T} )$ can be expressed as $ \bm{B}(t) = \sum_{k=1}^\infty b_k \varphi_k(t) $
and ${\mathcal{K}} \bm{B}(s) = \sum_k \lambda_k b_k  \varphi_k(s) $.
The condition in Proposition \ref{prop:invalid IV identification} requires that $\lambda_k >0$ for all $k \geq 1$.
Consequently, if ${\mathcal{K}} \bm{B}(s) =0$, it implies $b_k=0$ for all $k$ and then $\bm{B}(s) =0$.
 

We present several examples of kernel functions that satisfy our assumptions. The  {Ornstein-Uhlenbeck kernel} \( K(s,t) = \exp(-|t - s|) \) and the  {Brownian motion kernel} \( K(s,t) = \min(s, t) \) both meet the required conditions. In addition, a range of reproducing kernels—such as polynomial splines, periodic splines, thin-plate splines, exponential splines, and logistic splines \citep{wang2011smoothing}—are injective mappings. 
These kernels also exhibit appropriate eigenvalue decay behaviors: for polynomial decay kernels, \( \lambda_k \asymp k^{-2r} \), while for exponential decay kernels, \( \lambda_k \asymp \exp(-\alpha k^r) \), both satisfying our theoretical assumptions.

\subsection{Identifiability Condition for linear $f$} 
We investigate the linear case when $f$ is linear in $X$. In this case, model \eqref{FLSEM} simplifies to 
\begin{equation}
\label{eq:FLSEM_linear} 
Y_i  = \sum_{\ell=1}^{p} X_{i\ell}\beta_\ell+ \int_{\mathcal T} \bm{Z}_i(t)\bm{B}(t) dt + \epsilon_{i}~~~~\mbox{and}~~~~    
\bm{Z}_i(t) = \sum_{\ell=1}^{p} X_{i\ell} C_\ell (t) +  \bm{E}_{i}(t).
\end{equation} 
Although there is one exposure, identifying the causal effect 
  $\bm{B}(t)$ remains challenging due to a limited number of instruments. 
For simplicity, we consider valid IVs such that 
$Y_i  =  \int_{\mathcal T} \bm{Z}_i(t)\bm{B}(t) dt + \epsilon_{i}$
and
$\bm{Z}_i(t) = X^\top_{i}\bm{C}(t) + \bm{E}_{i}(t) = \widetilde{\bm{Z} }_i(t) + \bm{E}_{i}(t)$.
Projecting  $\bm{Z}(t)$ onto the space of $X$ leads to 
 $      
\widetilde{\bm{Z} }  = X^\top \bm{C}(t).
$
It follows that
\begin{eqnarray*}
	E[  Y | X  ] = X^\top \int_{\mathcal T} C(t) \bm{B}(t) dt +0 = \int_{\mathcal T}  \widetilde{\bm{Z} } (t)  \bm{B}(t) dt.
\end{eqnarray*}
Here, the linear operator takes the form ${\mathcal{K} } \bm{B}(t) = \int_\mathcal{T} E\{ \widetilde{\bm{Z} } (s)  \widetilde{\bm{Z} } (t) \} \bm{B}(s) dt  $. A naive identifiable condition of $\bm{B}(s)$ hinges on  $\mathcal{N} ( { \mathcal{K} } )=\{0\}$.
However,  $\widetilde{\bm{Z} }(t) $ 
 typically contains less information compared to $\bm{Z}(t)$, making it challenging to satisfy the injectivity condition. 
Specifically, $\mathcal{K}\bm{B}(s)$ has the form
$$
\mathcal{K}\bm{B}(s) =  \int_{\mathcal{T}}   \bm{C}(t)^\top E(X X^\top) \bm{C}(s) \bm{B}(t) dt  =    \int_{\mathcal{T}}   \bm{C}(t)^\top E(X X^\top) \bm{B}(t) dt  \bm{C}(s). 
$$
Thus, the null space of $\mathcal{K}$ is $\mathcal{N} ( \mathcal{K} )=\{\bm{B}: \int_{\mathcal{T}}   \bm{C}(t)^\top E(X X^\top) \bm{B}(t) dt =0\}$.
Then $\bm{B}(t) $ is not identifiable for $\bm{B} \in \mathcal{N} ( \mathcal{K} )$ because the solutions to the functional normal equation are not unique.
To address this issue, methods like the Moore-Penrose generalized inverse \citep[see e.g., Definition 3.5.7 in][]{hsing2015theoretical} and regularization techniques such as Tikhonov regularization are popular \citep{engl1996regularization}.


While these methods guarantee a unique solution, they recover only the component of $\bm{B}(t)$ that lies in the subspace $\mathcal{N}(\mathcal{K})^\perp$, the orthogonal complement of the null space of $\mathcal{K}$. The dimension of $\mathcal{N}(\mathcal{K})$ is inversely related to the number of valid instruments: as the number of valid instruments increases, the null space shrinks, improving the identifiability of $\bm{B}(t)$. This highlights how richer instrumental information enhances the estimation of the leading components of the causal effect $\bm{B}(t)$ under appropriate conditions.

Before detailing the identifiable conditions, we first establish some mild assumptions regarding the functional forms of $\bm{B}(t)$ and $\bm{C}(t)$. Assume that
$\bm{B}(t) = \sum_{k=1}^\infty b_k \varphi_k(t)$ and $ \bm{C} (t) = \sum_{k=1}^\infty \bm{c}_ k \varphi_k(t) $, where $\{\varphi_k(t)\}$ represents an orthogonal basis.  
Model \eqref{eq:FLSEM_linear} implies
$$
E[ X_i \{Y_i - X_i^\top \beta- \int_{\mathcal T} \bm{Z}_i(t)\bm{B}(t) dt  \}]=0~~\mbox{and} \quad
E[X_i \{ \bm{Z}_i(t) - X_i^\top \bm{C}(t)  \}]=0.
$$
Letting $\Gamma^*=(\Gamma_\ell^*) = E( X_i X_i^\top )^{-1} E(X_i Y_i) $, we have that
\begin{eqnarray}
	\label{eq:decomposition_identification}
\Gamma^* = \beta+ \int_{\mathcal T} C(t) \bm{B}(t) dt = \beta+  \sum_{k=1}^\infty b_k \bm{c}_ k \approx 
\beta+  \sum_{k=1}^R b_k \bm{c}_ k, 
\end{eqnarray}
where $R$ is a positive integer. 
If $\int_{\mathcal T} C(t) \bm{B}(t) dt $ can be well approximated  by $ \sum_{k=1}^R b_k \bm{c}_ k$   \citep{he12000extending, cardot2003spline,shin2009partial}, then the leading coefficients $\{b_k\}_{k=1}^R$ can be identified under some conditions. If the $l$-the covariate $X_l$ is a valid instrument, then $\beta_l=0$. The estimates of $\Gamma^*$ and $\{\bm{c}_k\}$ can be obtained by examining the relationship between the outcome and the potential instruments, and that between the exposure and the potential instruments. Theorem \ref{thm:invalid_linear} below gives the identification condition at the presence of possible invalid instruments.

Recall that \(\mathcal{I}\) denotes the set of valid instruments, while \(\mathcal{C}\) represents the set of invalid instruments. Let the total number of instruments be \(L = p_\mathcal{I} + p_\mathcal{C}\), where \(p_\mathcal{I} = |\mathcal{I}|\) and \(p_\mathcal{C} = |\mathcal{C}|\) denote the number of valid and invalid instruments, respectively. 
We impose a mild assumption that the number of invalid instruments is bounded above: $p_\mathcal{C} < U$ for some $ U \leq L.$  
Importantly, this condition does not require knowledge of the exact identities or number of invalid instruments. Under this setup, the following theorem establishes that a meaningful coefficient can still be identified.

\begin{thm}
	\label{thm:invalid_linear}
	Suppose that $\int_{\mathcal T} \bm{C}(t) \bm{B}(t) dt $ can be approximated by 
	$\sum_{k=1}^R b_k \bm{c}_ k$ for some vectors $\bm{c}_ k$'s  such that $\int_{\mathcal T} \bm{C}(t) \bm{B}(t) dt = \sum_{k=1}^R b_k \bm{c}_ k + o(R^{-2r})$ for some $r\geq 1$.
	Consider all sets $S_m \subseteq \{1,\dots, L\}$ of size $|S_m| = L - U + 1 > R$ for all $\ m =1, \dots, M$ with the property that there exists a unique $\{b^{(m)}_{k}\}_{k=1}^R$ satisfying $ \sum_{k=1}^R b^{(m)}_{k} \bm{c}_ {\ell k} =  \Gamma_\ell^*$ for  $\ell \in  S_m$.
	There is a unique solution $\beta$ and $\{b_k\}_{k=1}^R$ to \eqref{eq:decomposition_identification} if and only if the subspace-restriction condition holds, that is $b^{(m)}_{k}=b^{(m^\prime)}_{k}$ for all 
	$m, m^\prime \in \{1, \dots, M\} $ and $k=1, \ldots, R$. 
\end{thm}

  Theorem~\ref{thm:invalid_linear} shows that while the full recovery of   \(\bm{B}(t)\) may not be possible in the presence of invalid instruments, its leading components can be consistently identified under suitable subspace restrictions. This extends the identifiability framework for a single endogenous variable \citep{kang2016instrumental,guo2018confidence} to the more complex and less explored setting with \(K\) endogenous variables. 
Specifically,  Theorem~\ref{thm:invalid_linear} states that   \(\{b_k\}_{k=1}^R\) is identifiable if {\it no two distinct instrument subsets of size at least \(L - U + 1\) yield the same coefficient estimates on any common set of \(R\) instruments while producing inconsistent results when combined.}  
This condition ensures that valid instruments provide sufficient identifying power and that no subset of invalid instruments can generate indistinguishable projections.

We illustrate Theorem \ref{thm:invalid_linear} with two examples. (i) Suppose there are five relevant instruments with $L=5$, 
\(R = 2\), \(U = 3\), and \(\Gamma^* = (2,3,3,8,5)^T\). Let 
\((c_{11}, \dots, c_{51}) = (1,1,2,1,2)\) and \((c_{12}, \dots, c_{52}) = (1,2,1,1,3)\).
According to Theorem \ref{thm:invalid_linear}, we have \(M = 2\) and two instrument subsets 
\(S_1 = \{1,2,3\}\) and \(S_2 = \{2,3,5\}\). One can verify that \(b^{(1)}_{1} = 1\) and \(b^{(1)}_{2} = 1\) satisfy 
$\sum_{k=1}^2 b^{(1)}_{k} \, c_{\ell k} = \Gamma_{\ell}^*
\quad\text{for all}\; \ell \in S_1,
$
while \(b^{(2)}_{1} = 1\) and \(b^{(2)}_{2} = 1\) hold similarly for \(S_2\). As a result, there is a unique solution 
for \(\beta\) and \(\{b_k\}_{k=1}^R\) with $p_\mathcal{I} =4$ and $p_\mathcal{C} = 1$.
(ii) Consider a different setting with \(R = 2\), \(U = 4\), \(\Gamma^* = (2,3,6,8,10)^T\), 
\((c_{11}, \dots, c_{51}) = (1,1,1,2,2)\), and \((c_{12}, \dots, c_{52}) = (1,2,3,2,3)\). 
In this case, \(M = 2\) and the subsets are \(S_1 = \{1,2\}\) and \(S_2 = \{4,5\}\) with 
\(b^{(1)}_{1} = 1\), \(b^{(1)}_{2} = 1\), \(b^{(2)}_{1} = 2\), and \(b^{(2)}_{2} = 2\). However, these values fail to satisfy 
the consistency requirement, so there is no unique solution for the system. We cannot distinguish between valid and invalid instruments.

While it may be difficult to check the condition in Theorem \ref{thm:invalid_linear} in general, the following corollary presents a sufficient condition that is easy to check in practice.

\begin{cor}
\label{cor:number of invalid IV}
Suppose the conditions in Theorem \ref{thm:invalid_linear} hold,
if the maximal number of invalid instruments $U \leq (L-R+1)/ 2$, there is a unique solution to \eqref{eq:decomposition_identification}.
\end{cor}



Corollary~\ref{cor:number of invalid IV} provides a computationally feasible criterion for verifying the identifiability conditions in Theorem~\ref{thm:invalid_linear}. The key requirement is that any set of \(R\) valid instruments must yield a unique solution for the coefficient vector \(\{b_k\}_{k=1}^R\), which amounts to uniquely solving the \(R\) equations  
$
\sum_{k=1}^R b^{(m)}_{k} \bm{c}_{lk} = \Gamma_l^*, \quad l = 1, \dots, R.
$   
Rather than exhaustively checking all subsets of instruments, the corollary simplifies verification by requiring that the number of invalid instruments be strictly less than \((L - R + 1)/2\). This generalizes the classical majority rule of \citet{kang2016instrumental} to the more complex setting involving endogenous functional variables. 
Notably, when \(R = 1\), the result reduces to the standard majority rule for scalar endogenous variable. 

In biomedical research, many genetic variants identified through Genome-Wide Association Studies (GWAS) are relevant to downstream analyses. While researchers may have a rough estimate of the maximum number of invalid instruments \(U\), the precise identities of these instruments are typically unknown. The effects of an endogenous functional exposure can still be identified, as long as the number of invalid instruments is less than \((L - R + 1)/2\).

 \section{Estimation}
 \label{sec:estimation}

We estimate the coefficient $\bm{B}(t)$ in \eqref{eq:FLSEM_linear} using a three‐step strategy, starting with a fast screening procedure designed for the \(p\gg n\) regime. In this initial step, we reduce the covariate set by applying two complementary filters: (i) sure independence screening on \((Y_i, X_i)\) \citep{fan2008sure}, and (ii) distance‐correlation screening on the functional pair \((Z_i(t), X_i)\) \citep{li2024partially}. Taking the union of these selected predictors produces a joint screening set that retains the most informative variables while ensuring that subsequent estimation remains both statistically sound and computationally feasible.

Second, we estimate   the fitted value of \(\bm{Z}_i(t)\) by fitting the model in \eqref{FLSEM} with separate controls for smoothness and sparsity. Building on \citet{mirshani2021adaptive}, we replace their elastic-net sparsity penalty with an \(L_0\) constraint, yielding the first function‐on‐scalar regression with exact sparsity. We solve $\widehat{\mathbf{C}}$  by minimizing 
\begin{eqnarray}
\label{eq:functional response loss}
	L_\lambda (\mathbf{C}) =
\frac{1}{2n}\bigl\|\mathbf{Z} - \mathbf{X}^\top \mathbf{C}(t)\bigr\|^2_{2}
+ \frac{\lambda_K}{2}\sum_{\ell=1}^p \bigl\|C_\ell\bigr\|^2_{\mathcal{K}}
+ \lambda_0 \,\bigl\|\mathbf{C}\bigr\|_0
\end{eqnarray} 
over $\mathbf{C}(t)\in\mathcal{H}(K)$, 
where \(\mathbf{C}(t)=(C_1(t),\dots,C_p(t))^\top\), 
$\bm{X}$ stacks the $n$ observations where the $i$-th row is $X_i^\top$ and $\bm{Z} = (\bm{Z}_1 (t_{1} ),  \dots,  \bm{Z}_1 (t_{m} ), \dots,  \bm{Z}_n (t_{1} ),  \dots,  \bm{Z}_n (t_{m} ) )$.
By the Representer Theorem \citep{wahba1990spline}, each estimator \(\widehat C_\ell(t)\) satisfies 
\[
\widehat C_\ell(t) = C_\ell^\top K(t), 
\quad
C_\ell\in\mathbb{R}^m,\;
K(t) = \bigl(K(t,t_1), \dots, K(t,t_m)\bigr)^\top. 
\]
Smoothness is enforced by the RKHS norm \(\|C_\ell\|_{\mathcal{K}}=\sqrt{C_\ell^\top \Sigma\,C_\ell}\), where $\Sigma=(\Sigma_{ii'}) = (K(t_i, t_{i'})  $, while the \(L_0\) term selects at most \(J\) nonzero functions. This finite‐dimensional formulation has been well studied in the RKHS literature \citep{zhang2022high}. 
Under the condition that $p \ll n$, the functional estimates have the following form with $\check{\mathbf{X}} = \mathbf{X} \otimes \Sigma $,
$$
\widehat{\mathbf{C}}=(\widehat C_1^\top, \dots, \widehat C_p^\top )^\top = ( \check{\mathbf{X}}^\top \check{\mathbf{X}} + nm \lambda_K \Sigma )^{-1} \check{\mathbf{X}}^\top   \bm{Z}.
$$



Building on the finite-dimensional representation of the functional coefficients, we propose the  {Functional Group Support Detection and Root Finding (FGSDAR)} algorithm (Algorithm~\ref{alg:estimation}) to efficiently optimize the loss function in \eqref{eq:functional response loss}. FGSDAR extends the support detection and root-finding framework of \cite{huang2018constructive} to the functional response setting. 
The algorithm employs a blockwise coordinate descent strategy, iteratively updating the active set (indices with nonzero coefficients) and the inactive set (zero coefficients). Estimation is performed only on the active set, thereby reducing model complexity and improving computational efficiency. 
For any index set \( A \subseteq \{1, 2, \ldots, p\} \) with size \( |A| \), let \( \bm{C}_A = \{ \bm{C}_\ell(t) : \ell \in A \} \), and similarly define \( \check{\mathbf{X}}_A \). Let \( \bm{C}|_A \) denote the vector where the \(\ell\)-th component is given by \( (\bm{C}|_A)_\ell = \bm{C}_\ell(t) \mathbbm{1}(\ell \in A) \), with \( \mathbbm{1}(\cdot) \) denoting the indicator function. Let \( \| \bm{C} \|_{J, \infty} \) represent the \(J\)-th largest \(L_2\)-norm among the coefficient functions. 
The validity of FGSDAR in identifying a local minimizer is established in Lemma~\ref{lem:minimizer} of the supplementary materials.

\begin{algorithm}[t] 
	\caption{Functional group support detection and root finding (FGSDAR)} 
	\label{alg:estimation}
	\begin{algorithmic}[1] 	
		\Require 
		An initial $\mathbf{C}^0$ and the group sparsity level $J$; set $k=0$.
		\State Calculate $d^0 = \check{\mathbf{X}}^\top ( \bm{Z} - \widetilde{\mathbf{X}} \mathbf{C}^0  )  - nm \lambda_{K}  \Sigma  \mathbf{C}^0 $ for a given smoothness parameter $\lambda_K$;
		\For{$k=0,1,2,\ldots$}
		\State $A_R^k = \{ \ell: \|\mathbf{C}_\ell^k+d_\ell^k\|^2 \geq \|\mathbf{C}^k + d^k\|^2_{J,\infty} \}, ~~ I_R^k = (A_R^k)^c$;
		\State $\mathbf{C}_{A_R^k}^{k+1} = ( \check{\mathbf{X}}_{A_R^k}^\top \check{\mathbf{X}}_{A_R^k} + nm \lambda_K \Sigma )^{-1}  \check{\mathbf{X}}_{A_R^k}^\top \bm{Z}$;
		\State  $d_{I_R^k}^{k+1} =  \check{\mathbf{X}}_{I_R^k}^\top ( \bm{Z}_R - \check{\mathbf{X}}_{A_R^k}  \mathbf{C}_{A_R^k}^k  )  - nm \lambda_K  \Sigma  \mathbf{C}_{I_R^k}^{k} $;
		\If{$A_R^{k+1} =A_R^k$} Stop.
		\Else $~k=k+1$; 
		\EndIf		
		\EndFor
		\State Select $\lambda_K$ by minimizing the generalized cross-validation (GCV) criterion;
		\Ensure 
		$\widehat{\mathbf{C}},~ A_R^k, ~\mbox{and}~I_R^k$. 
	\end{algorithmic} 
\end{algorithm}

Third,
 we write the outcome generating model as 
  \begin{eqnarray}
   Y_i  = \sum_{\ell=1}^{p} X_{i\ell}\beta_\ell+ \int_{\mathcal T} E\{\bm{Z}_i(t)|X_i\}\bm{B}(t) dt + \nu_{i} =
    \sum_{l=1}^{p} X_{i\ell}\beta_\ell+ \int_{\mathcal T} \widehat{\bm{Z}}_i(t)\bm{B}(t) dt + \hat\nu_{i},  \label{eq:plugged in form}
    \end{eqnarray} 
   where $\nu_i=Y_i-E(Y_i|X_i)= \int_{\mathcal T}[\bm{Z}_i(t)-  E\{\bm{Z}_i(t) |X_i\}]\bm{B}(t)dt+\epsilon_i$ satisfies 
   $E[\nu_i E\{\bm{Z}_i(t)|X_i\}]=0$. 
  Assume that $\bm{B}(t) $  resides in an RKHS $\mathcal{H}$. Consider the minimization problem
  \begin{eqnarray}
  \label{eq:loss_final}
  \min_{\beta \in \mathbb{R}^{p},  \bm{B}  \in \mathcal{H} } 
  \left [
  {1 \over 2n} \sum_{i=1}^n \Big\{  Y_i -  \Big( \sum_{l=1}^{p} X_{il}\beta_l+ \int_{\mathcal T} \widehat{\bm{Z}_i}(t)\bm{B}(t) dt  \Big)  \Big\}^2
  +\tau \|\beta \|_0 + {\lambda \over 2} \| \bm{B}  \|^2_\mathcal{K}
   \right],      
  \end{eqnarray}
 where $\tau$ and $\lambda$ are tuning parameters.
Replacing $ \bm{Z}_i(t) $ with $\widehat{\bm{Z}_i}(t)$ partials out the effect of $\widehat{\bm{Z}_i}(t)$ when selecting important variables for the outcome. 
The algorithm proposed in \cite{li2024partially} is adopted to solve \eqref{eq:loss_final}, deriving $\widehat{\mathcal{P} \cup \mathcal{C}  }=\{\ell: |\widehat\beta_\ell|>0 \}$ and $\widehat{ \bm{B} } (t) $.




The computational cost of solving problem \eqref{eq:functional response loss} grows prohibitively large when the sample size \(n\) or the number of grid points \(m\) increases. For example, in the UK Biobank imaging cohort, \(n\) can reach 100,000, and a single hippocampus image with \(100\times150\) resolution yields \(m = 15,000\) grid points. Evaluating \eqref{eq:functional response loss} over all samples and grid locations is therefore infeasible using standard optimization routines.

To address this challenge, we propose a region-based estimation algorithm tailored for large-scale problems. To handle the large sample size, we adopt a divide-and-conquer strategy \citep{zhang2015divide, lian2018divide, hong2022divide}, which partitions the dataset into \(\widetilde{n}\) subsets of equal size, each containing \(n / \widetilde{n}\) subsamples. This allows us to process smaller, tractable subproblems independently. In parallel, to accommodate the high spatial resolution of the imaging data, we introduce a sliding-window approach inspired by techniques in genetic data analysis \citep{hudson1988coalescent}. Specifically, the window moves across the imaging grid to generate overlapping subregions \(\{\mathcal{T}_1, \dots, \mathcal{T}_m\}\) that together cover the entire imaging domain \(\mathcal{T} = \bigcup_{k=1}^{m} \mathcal{T}_k\). 
 This dual strategy—dividing subjects and sliding over functional domains—yields significant computational savings and enables scalable analysis of large, high-dimensional imaging datasets.

After partitioning the data into \(\widetilde n\) subsamples and regions \(\{\mathcal{T}_k\}\), we apply the following distributed, region-based estimation:
(i) \textbf{Local Estimation.}  
    For each subsample \(\zeta=1,\dots,\widetilde n\) and region \(\mathcal{T}_k\), solve \eqref{eq:functional response loss} to obtain
   $       \widehat{\bm{C}}^{\zeta k}_\ell(t), \quad \ell=1,\dots,p,\; t\in\mathcal{D}_k.$  
  (ii) \textbf{Regional Aggregation.}  
    Within each region \(\mathcal{T}_k\), average across subsamples:
$
      \overline{\bm{C}}^k_\ell(t)
      = {\widetilde n}^{-1}\sum_{\zeta=1}^{\widetilde n} \widehat{\bm{C}}^{\zeta k}_\ell(t).
    $
  (iii)  \textbf{Function Reconstruction.}  
    For each subject \(i\) and \(t\in\mathcal{T}_k\), reconstruct
    $
      \widehat{\bm{Z}}_i(t)
      = \sum_{\ell=1}^p X_{i\ell}\,\overline{\bm{C}}^k_\ell(t).
    $ 
  (iv) \textbf{Overlap Synthesis.}  
    Since regions overlap, compute the final estimate at each grid point by averaging all predictions covering that point.


In addition to estimating the true \(\mathbf{B}^*(t)\), we also test whether this effect is identically zero.  Concretely, we consider the hypotheses  
$
H_0:\;\mathbf{B}^*(t)=0\quad\forall\,t$ v.s. $
H_1:\;\exists\,t\ \text{s.t.}\ \mathbf{B}^*(t)\neq 0.
$ 
Recall that by construction,  there exists \(\mathbf{F}^*(t)\in L^2(\mathcal{T})\) satisfying  
$
\mathbf{B}^*(t)=K^{1/2}\mathbf{F}^*(t).
$ 
Because \(K^{1/2}\) is invertible on its support, testing \(\mathbf{B}^*(t)\equiv0\) is equivalent to testing  
$ 
H_0:\;\mathbf{F}^*(t)\equiv 0\quad\forall\,t,
$ 
which we use as the basis for our inferential procedure.


Similar to \citet{cai2012minimax}, the solution to \eqref{eq:loss_final} admits the closed‐form
\begin{equation}\label{eq:F_t}
  \widehat{\mathbf{F}}(t)
  = {n}^{-1}\,(R_n + \lambda I)^{-1}\bigl(K^{1/2}\,\widehat{\mathbf{Z}}(t)\bigr)^\top \mathcal{M}_{\widehat A}\,Y,
\end{equation}
where 
$
R_n
= {n}^{-1}\,K^{1/2}\,\widehat{\mathbf{Z}}(s)^\top\,\mathcal{M}_{\widehat A}\,\widehat{\mathbf{Z}}(t)\,K^{1/2}
$ 
is the empirical analogue of the operator
$
R
= {n}^{-1}\,K^{1/2}\,E\bigl\{\widetilde{\mathbf{Z}}(s)^\top\,\mathcal{M}_{A^*}\,\widetilde{\mathbf{Z}}(t)\bigr\}\,K^{1/2},
$
with $\widetilde{\bm{Z} }  = X^\top \bm{C}(t)$,
and $I$ denotes the identity operator. The estimated active set of precision and confounder variables from \eqref{eq:plugged in form} is  denoted as $\widehat{A} = \widehat{\mathcal{P} \cup \mathcal{C}} = \{\ell: |\widehat{\beta}_\ell| > 0\}$, while $A^*$ represents the true useful control set. For any subset $A$, $\bm{X}_A$ consists of the columns $\{\ell: \ell \in A\}$ from $\bm{X}$. The operator
$\mathcal{M}_{A} = I - \bm{X}_{{A}} (\bm{X}_A^\top \bm{X}_A)^{-1} \bm{X}_A^\top$
projects out the influence of scalar covariates on the outcome.
To test \(H_0\!: \mathbf{F}^*(t)\equiv0\), we examine the squared norm of \((R_n+\lambda I)\widehat{\mathbf{F}}(t)\).  This yields the test statistic
\[
S_n
= {n\,\bigl\|\,(R_n+\lambda I)\widehat{\mathbf{F}}(t)\bigr\|_2^2}/{\widehat\sigma^{2}}
= {Y^\top\,\mathcal{M}_{\widehat A}\,\displaystyle\int_{\mathcal{T}}\!\!\int_{\mathcal{T}}
  \widehat{\mathbf{Z}}(s)^\top\,K(s,t)\,\widehat{\mathbf{Z}}(t)\,ds\,dt\,\mathcal{M}_{\widehat A}\,Y}
  /(n\,\widehat\sigma^2),
\]
where \(\widehat\sigma^2\) is the estimated variance of the outcome model residuals.

\section{Theoretical Properties}
\label{sec:theoretical}
We now turn to the theoretical analysis of our estimators and test statistic. 
Denote $\| \bm{C} \|_2^2 = \sum_{j=1}^{p} \| {C}_j \|_2^2 = \sum_{j=1}^{p} \int_\mathcal{T}C^2_j(t) dt$.
	Let $\delta^* = ( \delta^*_1, \ldots,  \delta^*_p )^\top \in (  {L_2} (\mathcal{T}) )^p$ be  
$
\delta_j^* \;=\;\arg\min_{\delta\in L^2(\mathcal{T})}
\;E\bigl(X_j - \!\!\int_{\mathcal{T}} \bm{Z}(t)\,\delta(t)\,dt\bigr)^2
$
for $j=1,\dots,p,$
corresponding to the projection onto the space spanned by \(\{\bm{Z}(t)\}_{t\in\mathcal{T}}\) for each covariate \(X_j\).
Using these population projections, we form the residualized covariates  
$
\widetilde X_{ij}
= X_{ij} - \int_{\mathcal{T}} \bm{Z}_i(t)\,\delta_j^*(t)dt$ and $
\widetilde{X}_i = (\widetilde X_{i1},\dots,\widetilde X_{ip})^\top,
$
and stack them across subjects into  
$
\widetilde{\mathbf{X}} = (\widetilde{X}_1^\top,\dots,\widetilde{X}_n^\top)^\top.
$

	\begin{asmp}
	\label{assum: J}
	The input sparsity levels for Algorithm \ref{alg:estimation} and for \eqref{eq:loss_final}
    satisfy $J \geq \max(J_y^*, J_z^*)$, 
	where $J_y^*$ and $J_z^*$ are the true sparsity levels for the scalar and imaging responses.
\end{asmp}


\begin{asmp}
	\label{assum: SRC}
	For the input sparsity levels $J$,
	$\bm{X}$ satisfies the sparse Riese condition \citep{huang2018constructive}  that $\widetilde{\mathbf{X}} ,\bm{X}  \sim \text{SRC}\{2J,c_-(2J), c_+(2J) \}$  
	for $A, B \subseteq S,$ $|A| \leq J,$ $|B| \leq J,$  and $A \cap B = \emptyset$ and $\forall u \neq 0 \in \mathbb{R}^{ |A| } $,  
	$
	0 < c_-(J) \leq  { \| \bm{X}_A u  \|_2^2  }/(n \| u \|_2^2) < c_+(J)$ and  
	$\theta_{J,J} \geq { \| \bm{X}^\top_B  \bm{X}_A  u  \|_2  }/(n \| u \|_2),$
where $c_-(J)$ is an increasing function of $J$, $c_+(J)$ is a decreasing function of $J$, 
	and $\theta_{J, J} =  \max\{( 1 - c_-(2J) ) ,  (  c_+(2J) -1 )\} $ is an increasing function of $J$.
\end{asmp}

\begin{asmp}
	\label{assum: 4-th order}
	For any function $ \bm{B}\in \mathcal{H} $, there exists some positive constant $c_1$ satisfying 
	$
	E \big\{  \int_0^1 \bm{Z} (t)  \bm{B}(t) dt \big\}^4 \leq c_1 \Big( E \big\{ 
	\int_0^1 \bm{Z} (t)  \bm{B}(t) dt 
	\big\}^2 \Big)^2
	$	
\end{asmp}

\begin{asmp}
	\label{assum: tilde x error}
	For $ j=1, \ldots, p$, $\widetilde{X}_{1j}, \ldots, \widetilde{X}_{nj}$ are independently and identically distributed with mean zero
	and $\sigma_x^2 = \max\{ Var( \widetilde{X}^2_{ij} ),  j=1, \ldots, p \}$ is finite.
\end{asmp}

\begin{asmp}
	\label{assum: conditional expectation}
	The coefficients $\delta_j^* \in \mathcal{H}$ for $j=1, \ldots, p$ and $ \delta_{\max}=\max_j{  \| \delta_j^*\|_\mathcal{H} }  < \infty$.
\end{asmp}

\begin{asmp}
	\label{assum: error}
	The random errors $\epsilon_1, \ldots, \epsilon_n$ are independently and identically distributed with mean zero, variance $\sigma^2$ and sub-Gaussian tails.
\end{asmp}

\begin{asmp}
	\label{assump:kernel assumption}
	There exists some constant $C_K > 0$ such that $\sup_t K(t,t) \leq C_K$.  
\end{asmp}


\begin{asmp}
\label{assump:error of imaging}
The errors $E_i(t)$ are identically independent and satisfy that $E\{E_i(t)\}=0$, 
$\sup_t E\{E_i(t)\}^2 < \sigma_E^2 < \infty$,  and the covariance function  Cov$(E_i(t), E_i(s))=G(s,t)$ for $s, t \in [0, 1]$
satisfies $0 <c_G \leq G(s,s) \leq C_G \leq \infty$.
\end{asmp}

Assumptions \ref{assum: J} and \ref{assum: SRC} mirror conditions in \citet{huang2018constructive}. In particular, Assumption \ref{assum: J} requires that our chosen sparsity level exceed the true model’s sparsity, ensuring all nonzero coefficients can be recovered. Assumption \ref{assum: SRC} imposes restricted eigenvalue–type bounds on the diagonal blocks of both \(\mathbf{X}^\top\mathbf{X}/n\) and \(\widetilde{\mathbf{X}}^\top\widetilde{\mathbf{X}}/n\), thereby controlling multicollinearity among predictors. 
Assumption \ref{assum: 4-th order}, standard in functional regression theory (e.g., \citet{yuan2010reproducing}), governs the fourth‐moment behavior of the functional covariates. Assumption \ref{assum: tilde x error} then bounds the projection error of scalar covariates onto the functional space, while Assumption \ref{assum: conditional expectation} requires that the projection functions \(\delta_j^*\) lie in the same RKHS as the true \(\mathbf{B}(t)\) 
(see also \citet{shin2009partial,li2020inference}). 
Finally, Assumption \ref{assum: error} is the usual sub‐Gaussian noise condition (\citet{huang2018constructive}), Assumption \ref{assump:kernel assumption} places regularity conditions on the RKHS kernel \citep{zhang2015divide}, and Assumption \ref{assump:error of imaging} controls the mean and variance structure of the imaging‐error process.

We begin by quantifying the approximation error for our solution sequence  of \(\mathbf{C}(t)\). Let \(\mathbf{C}^*(t)\) denote the true functional coefficient vector, and define two true active index sets: $A^*_R = \{j: \bm{C}_j(t) \neq 0\}$ and 
  $A^* = \{j: \beta_j \neq 0\}$ for the imaging and the scalar outcome.

\begin{thm}
	\label{thm:imaging_approximation_error}
Let \(J\) be the user–specified sparsity level in FGSDAR, and let \(A_R^{k+1}\) denote the active set estimated at iteration \(k+1\).  Under Assumptions \ref{assum: J}, \ref{assum: SRC}, \ref{assump:kernel assumption}, and \ref{assump:error of imaging}, and for a contraction constant \(\gamma<1\) defined in \eqref{eq:gamma_C_def} of the supplement, the following bounds hold:
\[
\bigl\|\mathbf{C}^*\bigr|_{A_R^*\setminus A_R^{k+1}}\bigr\|_2
\;\le\;\gamma^{k+1}\|\mathbf{C}^*\|_2
\;+\;\frac{\gamma}{1-\gamma}\,h(J,\lambda_K),
\quad 
\|\mathbf{C}^{k+1}-\mathbf{C}^*\|_2
\;\le\;b_1\,\gamma^k\|\mathbf{C}^*\|_2
\;+\;b_2\,h(J,\lambda_K),
\]
where
$
h(J,\lambda_K)
=\max_{\lvert A\rvert=J}
\Bigl\|(T_{nm}^{A}+\lambda_K I)^{-1}
\bigl((nm)^{-1}\sum_{i=1}^n\sum_{j=1}^m
E_i(t_j)\,K_{t_{ij}}X_{i,A}
+\lambda_K\,\mathbf{C}^*_{A}\bigr)\Bigr\|_2,
$ 
and \(T_{nm}^{A}:\mathcal{K}\to\mathcal{K}\) is the empirical operator
$
T_{nm}^{A}f
= (nm)^{-1}\sum_{i=1}^n\sum_{j=1}^m
\langle X_{i,A}^\top f,\,K_{t_{ij}}\rangle_{\mathcal K}
\,K_{t_{ij}}\,X_{i,A}^\top.
$ 
The constants \(b_1,\) and \(b_2\) are given in \eqref{eq:b_1_b_2} of the Supplement.  Moreover, for any \(\nu\in(0,1/2)\), 
$
h(J,\lambda_K)
\le\sigma_E\sqrt{J\log\!\bigl({p}/{\nu}\bigr)}
\;\sqrt{{\mathrm{tr}\bigl(K(K+\lambda_KI)^{-1}\bigr)}/(nm)
+\,{n}^{-1}} + J\sqrt{ \lambda_K/2} 
$
with probability at least \(1-2\nu\).
\end{thm}

Theorem~\ref{thm:imaging_approximation_error} establishes \(L_2\)–convergence of the FGSDAR iterates.  The first bound quantifies the remaining signal in \(A_R^*\) omitted from the active set after \(k+1\) iterations, while the second bound controls the estimation error 
\(\|\mathbf{C}^{k+1}-\mathbf{C}^*\|_2\), showing geometric decay at rate \(\gamma^k\) plus the residual term \(h(J,\lambda_K)\). 
The final estimation error is determined by \(h(J,\lambda_K)\).
These expressions delineate how the sparsity level \(J\), the sample dimensions \((n,m)\), and the penalty parameter \(\lambda_K\) jointly govern the residual approximation error.

Compared to existing function-on-scalar regression methods with high-dimensional covariates \citep{chen2016variable, barber2017function, cai2022robust}, which rely on predefined basis expansions and group Lasso-type penalties for variable selection, the proposed FGSDAR method adopts an $L_0$ penalty within the RKHS framework. Our theoretical analysis provides a general estimation error bound applicable to various RKHSs, which reduces to known results under specific conditions. For instance, if the kernel eigenvalues satisfy $\kappa_j \asymp j^{-2\alpha}$, then $\mathrm{tr}(K(K+\lambda_K I)^{-1}) = O(\lambda_K^{-1/(2\alpha)})$.
When $m$ is small, the optimal tuning is $\lambda_K = \left( {\log (p)}/(nmJ) \right)^{2\alpha / (2\alpha + 1)}$, which yields
$
h(J, \lambda_K) = O\left( J \left({\log (p)}/(nmJ) \right)^{\alpha / (2\alpha + 1)} \right).
$
For large $m$, it simplifies to $h(J, \lambda_K) = O\left( \sqrt{J \log (p) / n} \right)$, which recovers the existing results  \citep{barber2017function, parodi2018simultaneous} if $J=J_z^*$.

Second, we provide nonasymptotic error bounds for both the scalar and functional coefficient estimates in the outcome model. These results parallel those of \citet{li2024partially}, which analyzed a high‐dimensional, partially functional linear regression framework.

\begin{thm}
	\label{thm:consistency_revelant_controls}
    Under Assumptions \ref{assum: J}–\ref{assump:error of imaging}, and using independent splits to fit the exposure and outcome models, let \(\check\gamma<1\) be defined in terms of \(\theta_{J,J}\) and \(c_{-}(J)\).  Then for any \(\nu\in(0,1/5)\), with probability at least \(1-5\nu\), the following hold: 
	\bas
\|   \beta^* |_{A^* \backslash A^{k+1} }    \|_2 \leq \check{\gamma}^{k+1} \|\beta^* \|_2 + {\check{\gamma} \over (1 - \check{\gamma}) \theta_{J,J} } \varepsilon_1, \quad
\|    \beta^{k+1} - \beta^*  \|_2 \leq \big( 1 + {\theta_{J,J} \over c_{-}(J)} \big) \check{\gamma}^k \|\beta^* \|_2  + b \varepsilon_1
\eas
for some constant $b>0$,
 where  
$
\varepsilon_1
= \sqrt{{8J\,M_3\log(2p/\nu)}/{n}}
\;+\;\sigma_\epsilon\sqrt{{4J\log(2p/\nu)}/{n}},
$ 
and \(M_3,M_4\) are defined in \eqref{eq:M_3_M_4} of the supplement. 
  Next,  define $T=K^{1/2}CK^{1/2}$ and $C(s,t) = \Mean\{ \widetilde {\bm{Z}}(s) \widetilde {\bm{Z}}(t) \}$. 
  If $\widetilde{\bm{Z}}^*$ is an independent copy of $\widetilde{\bm{Z}}$ and $E^*$ is the expectation taken over $\widetilde{\bm{Z}}^*$, then 
$
\Mean^* \langle  \widehat{\bm{B}} - \bm{B}^*, \widetilde{\bm{Z}}^*   \rangle^2 =O( \lambda +  J \varepsilon^2_1  +  { tr( T(T + \lambda I)^{-1} )/n }  ) $.  Moreover, if \(B^*\) is smoother than the RKHS \(\mathcal K\), then 
$\| \widehat{\bm{B}} - \bm{B}^* \|^2_{\mathcal{K}} = O( \lambda + J \varepsilon^2_1 + { tr( T(T + \lambda I)^{-2} )/n }  )$.
\end{thm}

Theorem~\ref{thm:consistency_revelant_controls} quantifies the estimation error for the outcome‐model coefficients at each FGSDAR iteration. In particular, the error in estimating the scalar controls \(\beta\) scales as
$
O\!\bigl(\sqrt{J\,\log(p)/n}\bigr),
$ 
which attains the minimax optimal rate when \(J=J_y^*\) \citep{raskutti2011minimax}. For the functional causal effect \(\mathbf{B}(t)\), we operate within a general RKHS framework and derive a prediction‐error bound involving the term \( tr\bigl(T(T+\lambda I)^{-1}\bigr)\), often interpreted as the “effective dimension” in learning theory \citep{zhang2005learning}. Under additional smoothness conditions on \(\mathbf{B}(t)\), this prediction guarantee can be strengthened to convergence in the RKHS norm, echoing the classic result of \citet{yuan2010reproducing} that prediction in functional linear models is generally more tractable than full functional estimation.



Before stating the null‐limit theorem for the test statistic \(S_n\), we introduce some spectral notation. By Mercer’s Theorem, the integral operator \(R\) admits the decomposition  $
R(s,t)
=\sum_{j=1}^\infty \widetilde s_j\,\varphi_j(s)\,\varphi_j(t), $  
where \(\widetilde s_1>\widetilde s_2>\cdots\) are its eigenvalues and \(\{\varphi_j\}\) are the corresponding orthonormal eigenfunctions.   
With this notation in hand, we can now characterize the limiting null distribution of \(S_n\).

\begin{thm}
    \label{thm:testing}
Suppose Assumptions \ref{assum: J}–\ref{assump:error of imaging} hold.  Under the null hypothesis \(\mathbf{B}^*(t)\equiv0\), the test statistic \(S_n\) converges in distribution to
$
\sum_{j=1}^\infty \widetilde s_j\,x_j^2,
$ 
where \(\{\widetilde s_j\}\) are the eigenvalues of the operator \(R\) and the \(x_j\) are independent \(N(0,1)\) random variables. 
\end{thm}

Theorem \ref{thm:testing} shows that under \(H_0\),
 $ S_n $ converges to 
a weighted sum of independent \(\chi^2_1\) variables in distribution. For practical use, we approximate this null distribution by a scaled chi‐square \(\kappa\,\chi^2_\zeta\) via Welch–Satterthwaite moment matching: 
$\zeta = ( \sum_j \widetilde s_j )^2 / \sum_j \widetilde s^2_j$ and  $\kappa = \sum_j \widetilde s^2_j /\sum_j \widetilde s_j$.
Rather than estimating each \(\widetilde s_j\) directly, we use the sample operator \(R_n\) to compute
$\widehat \zeta = tr(R_n)^2 / tr(R_n^2)$ and 
$\widehat \kappa = tr(R_n^2) / tr(R_n)$. 
When there are no scalar confounders (\(A^*=\emptyset\)), the operator \(R\) simplifies to
$R=K^{1/2} E\{ {\bm{Z}}_i(s) {\bm{Z}}_i(t) \} K^{1/2}$ 
 \citep{cai2011optimal,cai2022robust,li2024partially}. In that case, our test statistic also applies directly to the standard partial functional linear model by replacing \(\widehat{\bm{Z}}(t)\) with \( \bm{Z}(t)\) in \eqref{eq:F_t}.

\section{Simulation Studies}
\label{sec:simulation}

We assess the finite‐sample performance of our estimator in four settings as follows.  

\begin{example}
	\label{ex:one_dimensional}
    In the first, we consider a one‐dimensional functional effect 
$B(t).$ 
    We simulate each subject's data as follows.
    The functional covariate combines five scalar predictors and two random coefficients as 
$\bm{Z}_i(t)
=\sum_{\ell=1}^5 X_{i\ell}\,C_\ell(t)
\;+\;\widetilde\xi_{i1}\,\varphi_1(t)
\;+\;\widetilde\xi_{i2}\,\varphi_2(t),
$
where 
the scalar covariates \(\mathbf X_i\) follow a multivariate normal distribution with an AR(1) correlation \(\Corr(X_{ij},X_{ik})=\rho_1^{|j-k|}\) and 
the coefficient functions \(C_1,\dots,C_5\) are
	\begin{eqnarray*}
		C_1(t) &=& 2t^2, \quad C_2(t) = \cos( 3 \pi t/2+ \pi/2 ), \quad C_3(t) = \sqrt{2} \sin(\pi t /2 )+ 3 \sqrt{2} \sin(3 \pi t /2) \\
		C_4(t) &=& 25 \exp(-t), \quad C_5(t) = 5+ 7 t.
	\end{eqnarray*}	
Moreover,  the functional effect is
$
B(t)=\sum_{k=1}^{10}4(-1)^{k+1}k^{-2}\,\varphi_k(t),
$ 
where the orthonormal basis functions are $
\varphi_{2k-1}(t)=\sqrt{2}\cos\bigl((2k-1)\pi t\bigr)~~\mbox{and}~~ 
\varphi_{2k}(t)=\sqrt{2}\sin\bigl((2k-1)\pi t\bigr)$ for 
  $ k=1, \ldots, 5.$ 
  	We only  observe $\bm{Z}_i(t_{ij})$ at 100 equally spaced points in $[0,1]$.
  We then generate the response
$
Y_i
=\sum_{k=1}^p X_{ik}\,\beta_k
+\int_0^1 Z_i(t)\,B(t)\,dt
+\epsilon_i,
$ 
using the true coefficient vector $\beta=(7, 0, 0, 0, 0, 5.5, 4, 3.5, 5, 4.5,  \underbrace{0,\ldots,0}_{p-10})^\top$. 
The random components \(\widetilde\xi_{i1}\), \(\widetilde\xi_{i2}\), and \(\epsilon_i\) are jointly Gaussian with variances \(\Var(\widetilde\xi_{i1})=1\), \(\Var(\widetilde\xi_{i2})=0.64\), and \(\Var(\epsilon_i)=1\).  Their covariances are \(\Cov(\widetilde\xi_{i1},\widetilde\xi_{i2})=0\), \(\Cov(\widetilde\xi_{i1},\epsilon_i)=\rho_2\), and \(\Cov(\widetilde\xi_{i2},\epsilon_i)=0.8\,\rho_2\), ensuring \(\epsilon_i\) is correlated with the functional covariate.
Under this setup, we have   
	$\mathcal{C}=  \{1\}$, 
	  $\mathcal{P}=  \{6,7,8,9,10\}$, 
$\mathcal{I} =  \{ 2,3, 4,5\}$,  and the set of irrelevant variables $\mathcal{S} =  \{11,\ldots, p\}$.

 We conducted simulations with \(n=200\) observations under two scenarios for the scalar coefficient dimension: a low-dimensional case (\(p=20\)) and a high-dimensional case (\(p=500\)). In each scenario, we compared our  estimator to a baseline method that ignores endogeneity—denoted “PFLR”—which directly fits the partially functional linear model.

All results are averaged over 100 Monte Carlo replicates, implemented in R 3.6.0 on a Linux server (Intel Xeon E5-2640 v4 @ 2.40 GHz, 125 GB RAM). For each method, we evaluate $ 
\mathrm{MSE}_\beta = \|\widehat\beta - \beta\|_2^2$ and $ 
\mathrm{MSE}_{\mathbf B} = \|\widehat{\mathbf B} - \mathbf B\|_2^2$ 
to assess estimation accuracy of the scalar coefficients \(\beta\) and the functional effect \(\mathbf B\). We also record variable‐selection performance via 
 $\mbox{FZ}_Z $ 
and  $ \mbox{FZ}_Y,$  
 which are the  number of false zeros among scalar predictors for the functional covariate and the outcome, respectively, and 
  $\mbox{FN}_Z$ and $ \mbox{FN}_Y$, 
  which are the number of false nonzeros for those same sets. 
Finally, predictive accuracy is measured by the mean squared error (PMSE) on an independent test set of 200 observations.

 Table~\ref{table:MSEs_betaS} reports variable‐selection performance and estimation accuracy for the outcome model. When endogeneity is mild (small \(\rho_2\)), our method and the naive PFLR estimator perform similarly. As \(\rho_2\) increases, however, PFLR’s errors grow substantially, whereas our method maintains low mean squared errors for both scalar and functional coefficients.
In particular, our approach achieves far more reliable variable selection—fewer false nonzeros and essentially zero false exclusions of true predictors—demonstrating both higher accuracy and greater stability than PFLR.

Table~\ref{table:MSEs_Cs} presents the functional exposure‐model errors under varying covariate‐correlation settings. Once again, our estimator consistently delivers accurate surface estimates. 
Results for \(p=500\) (see Section~\ref{sec:additional_simu} of the Supplement) mirror these findings: in high‐dimensional settings with endogeneity, our method substantially outperforms PFLR while still providing precise estimates for both the outcome and exposure models.

	\begin{table}
		\begin{center}
			\caption{\small Monte Carlo averages and empirical standard errors in parentheses in Example \ref{ex:one_dimensional}.}
			\label{table:MSEs_betaS}
		\tabcolsep 1pt
			\renewcommand{\arraystretch}{0.6}
			\begin{tabular}{ccccccccccccc}
				\hline
				\hline
				$\rho_1$&      $\rho_2$ &   & FZ$_Z$  &  FN$_Z$   &  FZ$_Y$  &  FN$_Y$  & MSE$_{\bm{B} }$  & MSE$_\beta$  \\
				\hline 
0.3	&	0	&	FLSEM	&	0.0(0.0) & 	0.0(0.0) & 	0.0(0.0) & 	0.393(0.486) & 	0.009(0.004) & 	0.050(0.022) & 	\\
	&		&	PFLM	&	- & 	- & 	0.0(0.0) & 	2.538(1.291) & 	0.007(0.002) & 	0.077(0.028) & 	\\
\cline{2-9}												
	&	0.2	&	FLSEM	&	0.0(0.0) & 	0.0(0.0) & 	0.0(0.0) & 	0.345(0.472) & 	0.008(0.003) & 	0.051(0.026) & 	\\
	&		&	PFLM	&	- & 	-  & 	0.0(0.0) & 	3.178(0.999) & 	0.036(0.017) & 	0.130(0.053) & 	\\
\cline{2-9}												
	&	0.5	&	FLSEM	&	0.0(0.0) & 	0.0(0.0) & 	0.0(0.0) & 	0.332(0.468) & 	0.008(0.003) & 	0.065(0.029) & 	\\
	&		&	PFLM	&	- & 	-  & 	0.0(0.0) & 	3.894(0.444) & 	0.486(0.126) & 	3.618(1.495) & 	\\
\cline{2-9}												
	&	0.7	&	FLSEM	&	0.0(0.0) & 	0.0(0.0) & 	0.0(0.0) & 	0.338(0.470) & 	0.008(0.003) & 	0.073(0.031) & 	\\
	&		&	PFLM	&	- & 	-  & 	0.0(0.0) & 	3.961(0.396) & 	1.166(0.130) & 	9.634(1.079) & 	\\
\hline												
0.5	&	0	&	FLSEM	&	0.0(0.0) & 	0.198(0.397) & 	0.0(0.0) & 	0.399(0.488) & 	0.010(0.005) & 	0.060(0.027) & 	\\
	&		&	PFLM	&	- & 	-  & 	0.0(0.0) & 	2.339(1.266) & 	0.006(0.002) & 	0.091(0.035) & 	\\
\cline{2-9}												
	&	0.2	&	FLSEM	&	0.0(0.0) & 	0.176(0.379) & 	0.0(0.0) & 	0.292(0.452) & 	0.009(0.005) & 	0.065(0.035) & 	\\
	&		&	PFLM	&	- & 	-  & 	0.0(0.0) & 	3.080(1.059) & 	0.035(0.015) & 	0.148(0.056) & 	\\
\cline{2-9}												
	&	0.5	&	FLSEM	&	0.0(0.0) & 	0.213(0.408) & 	0.0(0.0) & 	0.309(0.459) & 	0.010(0.005) & 	0.076(0.037) & 	\\
	&		&	PFLM	&	- & 	-  & 	0.0(0.0) & 	3.855(0.490) & 	0.474(0.132) & 	3.531(1.545) & 	\\
\cline{2-9}												
	&	0.7	&	FLSEM	&	0.0(0.0) & 	0.216(0.410) & 	0.0(0.0) & 	0.281(0.446) & 	0.010(0.005) & 	0.087(0.044) & 	\\
	&		&	PFLM	&	- & 	-  & 	0.0(0.0) & 	3.961(0.396) & 	1.166(0.130) & 	9.622(1.075) & 	\\
				\hline \hline  		
			\end{tabular}
		\end{center}
	\end{table}

 	\begin{table}
		\begin{center}
			\caption{\small Simulation results for the functional response of Monte Carlo averages and empirical standard errors in parentheses in Example \ref{ex:one_dimensional}.}
			\label{table:MSEs_Cs}
			
			\vspace{-0.2in}
			\tabcolsep 4pt
			\renewcommand{\arraystretch}{0.6}
			\begin{tabular}{ccccccccccccc}
				\hline
				\hline
				$\rho_1$  &  $\rho_2$   & MSE$_{C_1}$ & MSE$_{C_2}$ & MSE$_{C_3}$ & MSE$_{C_4}$ & MSE$_{C_5}$ & PMSE$_{Z}$\\
				\hline 
0.3	&	0	&	0.005(0.004) & 	0.006(0.004) & 	0.006(0.004) & 	0.005(0.003) & 	0.005(0.004) & 	1.621(0.196) & 	\\
	&	0.2	&	0.005(0.004) & 	0.005(0.004) & 	0.005(0.004) & 	0.005(0.003) & 	0.005(0.003) & 	1.613(0.188) & 	\\
	&	0.5	&	0.005(0.004) & 	0.005(0.004) & 	0.005(0.003) & 	0.006(0.004) & 	0.005(0.004) & 	1.612(0.190) & 	\\
	&	0.7	&	0.005(0.003) & 	0.005(0.004) & 	0.006(0.004) & 	0.005(0.004) & 	0.005(0.004) & 	1.606(0.187) & 	\\
\hline										
0.5	&	0	&	0.006(0.004) & 	0.007(0.005) & 	0.007(0.005) & 	0.009(0.007) & 	0.005(0.004) & 	1.625(0.195) & 	\\
	&	0.2	&	0.006(0.004) & 	0.007(0.005) & 	0.008(0.006) & 	0.008(0.006) & 	0.006(0.005) & 	1.622(0.188) & 	\\
	&	0.5	&	0.005(0.004) & 	0.007(0.005) & 	0.008(0.006) & 	0.009(0.006) & 	0.006(0.004) & 	1.617(0.191) & 	\\
	&	0.7	&	0.005(0.004) & 	0.007(0.005) & 	0.008(0.006) & 	0.008(0.005) & 	0.006(0.005) & 	1.611(0.187) & 	\\
				\hline \hline  		
			\end{tabular}
		\end{center}
	\end{table}

\end{example}


\begin{example}
	\label{ex:2D}
    In the second, we extend to a two‐dimensional surface 
$\bm{B}(t_1, t_2)$.  
We generate 
two-dimensional $\bm{Z}_i(t_1,t_2)$ as 
$\bm{Z}_i(t_1,t_2) = \sum_{l=1}^5 X_{ij} C_l (t_1,t_2)+  \widetilde \xi_{i 1} \varphi_1(t_1,t_2) + \widetilde \xi_{i 2} \varphi_2(t_1,t_2)$ and 
set $\bm{B}(t_1,t_2) =  \exp(-(t_1-t_2))$, where   $C_1(t) = 2(t_1^2 +t_2^2),$ $  C_2(t) = 3 \cos( \pi t_1/2 )\cos( \pi t_2/2 ),$  
$$  
	C_3(t) = {\sqrt{2} \over 2} \{\sin(\pi t _1/2 )+ 3 \sin(3 \pi t_2 /2)\}, 
	C_4(t) =\exp(-(t_1-t_2)), \quad C_5(t) = 2+ t_1+t_2.
$$ 
The random bases are \(\varphi_1=1.588\sin(\pi t_1)\) and  \(\varphi_2=2.157\{\cos(\pi t_2)-0.039\}\).  Scalar covariates and errors are generated as in Example~\ref{ex:one_dimensional}.  We observe \(Z_i\) on a \(100\times150\) grid over \([0,1]^2\) and set \(n=400\) and \(p=20\), with true 
\(\beta=(2,0,0,0,0,5.5,4,3.5,5,4.5,0,\dots,0)^\top\).  
As before,   we have   
	$\mathcal{C}=  \{1\}$, 
	  $\mathcal{P}=  \{6,7,8,9,10\}$, 
$\mathcal{I} =  \{ 2,3, 4,5\}$,  and   $\mathcal{S} =  \{11,\ldots, p\}$.


Table \ref{table:MSEs_2D} summarizes selection accuracy, estimation, and prediction errors for these   simulations. Similar to Example \ref{ex:one_dimensional},  our method outperforms PFLR across all settings. Table \ref{table:MSEs_Cs_2D} reports functional‐model errors under various correlation structures, confirming that FGSDAR yields accurate surface estimates in high‐dimensional imaging scenarios. 

	\begin{table}
	\begin{center}
		\caption{\small Simulation results of Monte Carlo averages for the outcome in Example \ref{ex:2D}.}
		\label{table:MSEs_2D}
		
		\vspace{0.8ex}
		\tabcolsep 2pt
		\renewcommand{\arraystretch}{0.6}
		\begin{tabular}{ccccccccccccc}
			\hline
			\hline
			$\rho_1$&      $\rho_2$ &   & FZ$_Z$  &  FN$_Z$   &  FZ$_Y$  &  FN$_Y$  & MSE$_B$ & MSE$_\beta$ \\
			\hline 
0.3	&	0	&	FLSEM	&	0.0(0.0) & 	7.120(1.996) & 	0.0(0.0) & 	0.080(0.274) & 	0.049(0.014) & 	0.027(0.019) & 	\\
&		&	PFLM	&		- & 	-& 	0.0(0.0) & 	2.300(1.216) & 	0.053(0.021) & 	0.063(0.036) & 	\\
\cline{2-9}												
&	0.2	&	FLSEM	&	0.0(0.0) & 	6.960(2.194) & 	0.0(0.0) & 	0.040(0.198) & 	0.051(0.016) & 	0.033(0.016) & 	\\
&		&	PFLM	&		- & 	- & 	0.0(0.0) & 	3.380(0.830) & 	0.235(0.074) & 	0.579(0.238) & 	\\
\cline{2-9}												
&	0.5	&	FLSEM	&	0.0(0.0) & 	7.380(2.108) & 	0.0(0.0) & 	0.120(0.385) & 	0.047(0.015) & 	0.037(0.026) & 	\\
&		&	PFLM	&		- & 	-& 	0.0(0.0) & 	3.960(0.198) & 	0.644(0.123) & 	3.296(0.524) & 	\\
\cline{2-9}												
&	0.7	&	FLSEM	&	0.0(0.0) & 	7.260(2.068) & 	0.0(0.0) & 	0.020(0.141) & 	0.048(0.013) & 	0.034(0.024) & 	\\
&		&	PFLM	&		- & 	-& 	0.0(0.0) & 	4.000(0.0) & 	0.945(0.015) & 	7.226(0.169) & 	\\
\hline												
0.5	&	0	&	FLSEM	&	0.0(0.0) & 	6.280(2.176) & 	0.0(0.0) & 	0.080(0.274) & 	0.046(0.012) & 	0.032(0.022) & 	\\
&		&	PFLM	&		- & 	-& 	0.0(0.0) & 	2.600(1.355) & 	0.059(0.028) & 	0.084(0.075) & 	\\
\cline{2-9}												
&	0.2	&	FLSEM	&	0.0(0.0) & 	7.220(2.122) & 	0.0(0.0) & 	0.060(0.314) & 	0.047(0.013) & 	0.034(0.031) & 	\\
&		&	PFLM	&		- & 	- & 	0.0(0.0) & 	3.560(0.733) & 	0.206(0.097) & 	0.552(0.288) & 	\\
\cline{2-9}												
&	0.5	&	FLSEM	&	0.0(0.0) & 	6.400(2.356) & 	0.0(0.0) & 	0.080(0.274) & 	0.042(0.010) & 	0.039(0.027) & 	\\
&		&	PFLM	&		- & 	-& 	0.0(0.0) & 	3.960(0.198) & 	0.616(0.089) & 	3.503(0.522) & 	\\
\cline{2-9}												
&	0.7	&	FLSEM	&	0.0(0.0) & 	7.180(2.077) & 	0.0(0.0) & 	0.080(0.274) & 	0.047(0.014) & 	0.037(0.026) & 	\\
&		&	PFLM	&		- & 	- & 	0.0(0.0) & 	4.000(0.0) & 	0.944(0.014) & 	7.217(0.172) & 	\\
			\hline \hline  		
		\end{tabular}
	\end{center}
\end{table}

	\begin{table}
	\begin{center}
		\caption{\small Estimation results for the two dimensional functional response in Example \ref{ex:2D}.}
		\label{table:MSEs_Cs_2D}
		
		\vspace{-0.2in}
		\tabcolsep 4pt
		\renewcommand{\arraystretch}{0.6}
		\begin{tabular}{ccccccccccccc}
			\hline
			\hline
$\rho_1$	&	$\rho_2$   &   MSE$_{C_1}$ & MSE$_{C_2}$ & MSE$_{C_3}$ & MSE$_{C_4}$ & MSE$_{C_5}$  &  PMSE$_Z$ \\
			\hline 
0.3	&	0	&	0.138(0.023) & 	0.139(0.024) & 	0.333(0.033) & 	0.068(0.028) & 	0.058(0.024) & 	3.381(0.159) & 	\\
&	0.2	&	0.132(0.021) & 	0.138(0.025) & 	0.334(0.041) & 	0.073(0.028) & 	0.070(0.044) & 	3.394(0.167) & 	\\
&	0.5	&	0.138(0.021) & 	0.140(0.028) & 	0.338(0.046) & 	0.068(0.018) & 	0.068(0.038) & 	3.389(0.159) & 	\\
&	0.7	&	0.137(0.025) & 	0.138(0.020) & 	0.336(0.039) & 	0.069(0.024) & 	0.063(0.029) & 	3.399(0.157) & 	\\
\cline{2-8}										
0.5	&	0	&	0.139(0.030) & 	0.136(0.025) & 	0.336(0.035) & 	0.067(0.023) & 	0.068(0.039) & 	3.410(0.175) & 	\\
&	0.2	&	0.135(0.023) & 	0.136(0.020) & 	0.332(0.029) & 	0.066(0.023) & 	0.070(0.041) & 	3.392(0.150) & 	\\
&	0.5	&	0.138(0.027) & 	0.141(0.018) & 	0.330(0.032) & 	0.074(0.027) & 	0.069(0.038) & 	3.393(0.129) & 	\\
&	0.7	&	0.136(0.026) & 	0.143(0.023) & 	0.329(0.031) & 	0.074(0.028) & 	0.070(0.048) & 	3.377(0.150) & 	\\
			\hline \hline  		
		\end{tabular}
	\end{center}
\end{table}
\end{example}

\begin{example}
    \label{ex:divide and conquer}



In the third, we evaluate the proposed divide‐and‐conquer procedures under the two‐dimensional design of Example~\ref{ex:2D}. As shown in Figure~\ref{fig:divide_y_z}, the divide‐and‐conquer estimator closely matches the accuracy and predictive performance of the full‐sample fit. Due to page limits, detailed descriptions are deferred to Section~\ref{sec:additional_simu} of the supplement.



\end{example}

\begin{example}
In our fourth example, we evaluate the finite‐sample performance of the proposed hypothesis test. We retain the design of Example~\ref{ex:one_dimensional} but scale the true functional effect to
$
B(t)
= b\sum_{k=1}^5(-1)^{k+1}k^{-2}\varphi_k(t)$ with  
$ b\in\{0,0.04,0.08,0.12,0.16,0.20\}.$ 

Figure~\ref{fig:testing} reports empirical rejection rates at the 5\% significance level for our method and for a naive test that treats the functional covariate as observed without correcting for endogeneity, based on 1,000 simulations. When endogeneity is absent (\(\rho_2=0\)), the naive test correctly controls size. As \(\rho_2\) increases, however, its Type I error inflates substantially, making its results unreliable. In contrast, our test maintains size near 5\% across all \(\rho_2\) levels, and its power approaches 100\% as the signal strength \(b\) grows.

      \begin{figure}
 	\centering
 	\includegraphics[height=6cm, width=13cm]{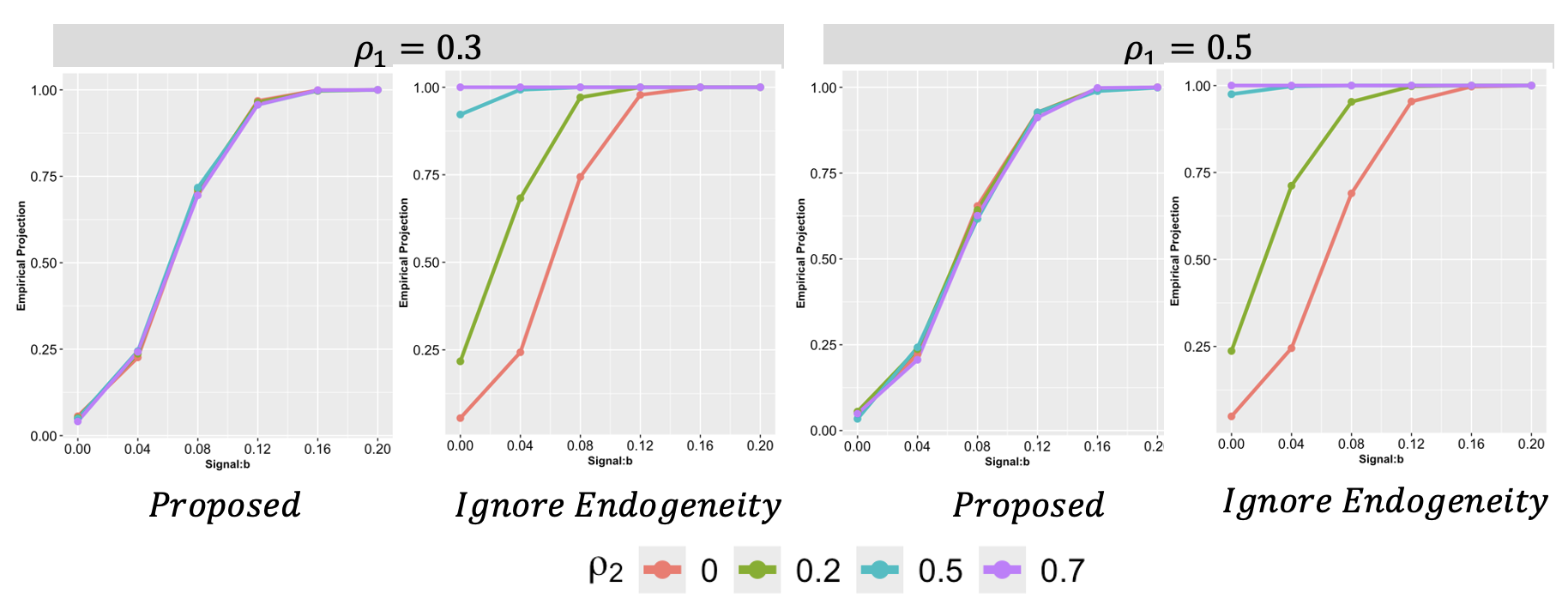}
 	\vspace{-0.1 in}
 	\caption{\small Empirical rejection of proposed testing method, ``Ignore endogeneity'' indicates using the observed functional data for testing.  }
 	\label{fig:testing}
 	\vspace{-0.1 in}
 \end{figure}
 
\end{example}

\section{Real Data Analysis of UK Biobank}
\label{sec:realData}

The UK Biobank is a large-scale, prospective cohort study that recruited over 500,000 participants aged 40–69 from across the United Kingdom between 2006 and 2010 \citep{sudlow2015uk}. One of its primary objectives is to collect comprehensive health and biological data prior to disease onset \citep{elliott2018genome}, thereby enabling the investigation of early predictors of complex traits and diseases.

Among the rich set of variables available, fluid intelligence scores are derived from participants' performance on 13 problem-solving questions, representing a core component of cognitive function. Meanwhile, brain function is typically assessed through task-based fMRI, which captures dynamic patterns of neural activation. Prior research has highlighted associations between structural and functional brain imaging and intelligence \citep{cheng2020large}, with findings suggesting that individuals with larger brain volumes tend to perform better on intelligence tests \citep{mcdaniel2005big, rushton2009whole}.

The central goal of our data analysis is to investigate how genetic variations and brain functional networks jointly influence individual differences in fluid intelligence. By integrating single-nucleotide polymorphism (SNP) data with fMRI measurements, we aim to (i) identify neural activation patterns and brain regions correlated with fluid intelligence, (ii) uncover genetic variants associated with these traits, and (iii) determine shared SNPs that may serve as confounders or mediators. This integrative approach offers critical insights into the neurobiological and genetic underpinnings of human intelligence. Ultimately, such findings could inform the design of targeted educational interventions, improve strategies for supporting cognitive development, and enhance early detection of cognitive impairments.

After rigorous quality control of both imaging and phenotypic data, we constructed an analytical dataset of  
n=11,400 participants. This dataset includes a comprehensive set of covariates: demographic factors (age, sex, education level, body mass index, social deprivation index), lifestyle variables (smoking, physical activity, healthy diet, moderate alcohol consumption), medical history (hypertension, diabetes), imaging-related technical variables (site, volumetric scaling, head position, brain position), and interaction terms (age·sex, age², sex·age²). To correct for population stratification, we additionally include the top 40 principal components from genome-wide SNP data \citep{price2006principal}.
Resting-state functional connectivity fMRI serves as the imaging data that measures brain activity, summarized by 78 grid points. Each grid point represents the functional connectivity between a pair of brain networks or regions. These regions are grouped into 12 canonical resting-state networks, with detailed labels shown in Figure~\ref{fig:pheno_gram}(b).

Given the ultrahigh dimensionality of the SNP data—over 96 million variants—we apply a two-step sure independence screening (SIS) procedure \citep{fan2008sure}, separately tailored for both the outcome (fluid intelligence) and imaging (fMRI) models. Specifically, we control for covariates and rank SNPs based on their marginal absolute correlations with the fluid intelligence score, retaining the top 2,000. For the imaging model, we conduct marginal screening across all fMRI voxels and rank SNPs by the aggregated absolute correlation across pixels, again retaining the top 2,000 SNPs. All SNPs and continuous variables are standardized for comparability.

This screening strategy aligns with practices from genome-wide association studies (GWAS), which aim to detect associations between genetic variants and phenotypes at scale. Table \ref{table:common snp} in the supplement presents the SNPs commonly selected in both the fluid intelligence and fMRI screening processes. Notably, these SNPs cluster within the same linkage disequilibrium (LD) block on chromosome 15 \citep{berisa2016approximately}, suggesting a shared genetic influence on both brain function and cognition.

Figure \ref{fig:pheno_gram}(a) provides a genomic ideogram of SNPs selected by our proposed FLSEM and PFLM. FLSEM selects 158 SNPs, PFLM selects 163, with 126 SNPs overlapping. Most of the shared SNPs are located on chromosomes 2, 6, and 14, highlighting potential key genomic regions that contribute to individual differences in cognitive ability through both direct and indirect pathways. 

Building on the leading SNPs identified through the sure independence screening, our analysis proceeds in three main steps. First, we fit an imaging-on-scalar model to estimate the imaging variables as a function of the selected genetic variants. This step captures the component of brain activation explained by the SNPs. Second, using the fitted imaging values as covariates, we apply a variable selection model to the fluid intelligence scores to identify SNPs that directly influence cognition after accounting for imaging-based effects. Finally, we estimate the direct effect of the imaging variables on fluid intelligence, adjusting for demographic, clinical, and genetic confounders, and the selected precision variables.

Figure \ref{fig:pheno_gram}(b) illustrates the estimated effects of fMRI-based brain connectivity on fluid intelligence derived from both the proposed FLSEM and the PFLM. The results reveal distinct patterns across methods. Notably, both models identify positive associations between connectivity in the language and ventral multimodal networks and higher fluid intelligence scores. However, connectivity involving the visual network displays divergent effects between the two models, suggesting method-dependent sensitivity to specific brain regions.

To formally evaluate the role of brain function in cognitive performance, we test the association between fMRI-derived imaging exposures and fluid intelligence scores. The proposed testing procedure yields a $p$-value of 0.08, indicating marginal significance at the 10\% level. To further probe the generalizability of these associations, we extend our analysis to other cognitive phenotypes available in the UK Biobank, including: Reaction Time (assessing processing speed), Trail Making (executive function and cognitive flexibility), Numeric Memory (working memory), Matrix Completion (reasoning and pattern recognition), and Tower Rearranging (planning and problem-solving). 
Significant associations are found for Reaction Time ($p = 0.003$), Trail Making ($p = 0.001$), and Tower Rearranging ($p = 0.009$), while Matrix Completion is marginally significant ($p = 0.087$). In contrast, Numeric Memory ($p = 0.268$) shows no significant relationship with the imaging data. These findings underscore the heterogeneous influence of brain function on different aspects of cognition and emphasize the need for targeted, phenotype-specific modeling.

 \begin{figure}
 	\centering
 	\includegraphics[height=7cm, width=15cm]{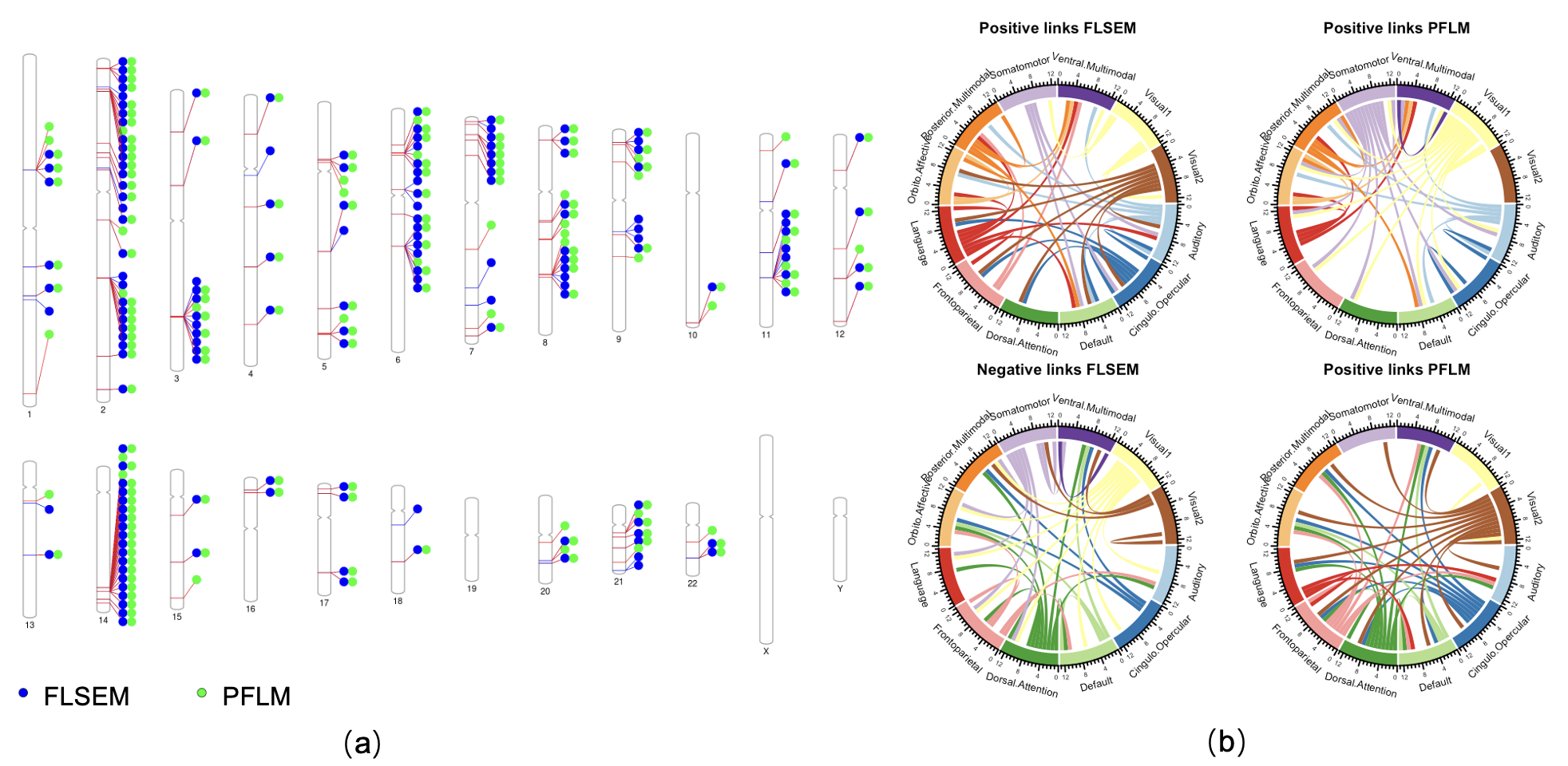}
 	\vspace{-0.1 in}
 	\caption{\small Panel (a) shows the positions of SNPs selected by FLSEM and PFLM; Panel (b) displays the positive and negative fMRI effect estimates on fluid intelligence from both methods.    }
 	\label{fig:pheno_gram}
 \end{figure}


\section{Discussion}


We introduced the FLSEM framework to uncover causal links among genetic, imaging, and clinical data under unobserved confounding.  Our \(L_0\)‐penalized, three‐step estimator—bolstered by the FGSDAR algorithm—simultaneously handles endogeneity, high‐dimensional scalars, and infinite‐dimensional functions.  We proved identification for mixed scalar–functional instruments, derived nonasymptotic error and testing guarantees, and demonstrated via simulations that our approach outperforms naive methods, especially under strong endogeneity.  An application to the UK Biobank confirms its practical value in large‐scale multimodal studies. 
The current framework can be extended to incorporate GWAS summary data for causal factor identification \citep{fang2025multivariate}, and to enable causal mediation analysis by modeling both direct and mediated treatment effects \citep{yang2024causal}.

\section{Acknowledgment}
This research has been conducted using the UK Biobank Resource under Application Number 22783, subject to a data transfer agreement. We thank the participants in the UKB study for their contribution and the research teams for the work in collecting, processing, and disseminating these datasets for analysis. We thank the University of North Carolina at Chapel Hill and the Research computing groups for providing computational resources and support that have contributed to the research results. 
  This paper was primarily completed during Dr. Ting Li's visit to UNC Chapel Hill. Dr. Fan's work  was partially supported by the National Science Foundation (NSF) grants DMS-2230795 and DMS-2230797. Dr. Zhu's work was partially supported by the Gillings Innovation Laboratory on generative AI and by grants from the National Institute on Aging (NIA) of the National Institutes of Health (NIH), including   1R01AG085581, and RF1AG082938, the National Institute of Mental Health (NIMH) grant 1R01MH136055,  and the NIH grants R01AR082684, and 1OT2OD038045-01. The content of this paper is solely the responsibility of the authors and does not necessarily represent the official views of these institutions.

 \section{Author Contributions} 
Ting Li proposed the statistical method, conducted the theoretical analysis and simulation studies, and drafted the manuscript. Ethan Fang assisted with the theoretical analysis and contributed to revising and editing the manuscript. Tengfei Li performed the real data analysis and contributed to the interpretation of empirical results. Hongtu Zhu, supervised the project, provided critical guidance throughout the research process, and revised the manuscript. All authors read and approved the final version of the manuscript.

\setlength{\baselineskip}{18pt}

\bibliographystyle{chicago} 
\bibliography{RefFLSEM}
\end{document}



\if1\blind
{
  \title{\bf  Causal Inference in Biomedical Imaging with  Functional Linear Structural Equation Models}
    \author{Ting Li
    \hspace{.2cm}\\
    School of  Statistics and Data Science,  \\  Shanghai University of Finance and Economics, Shanghai, China\\
   Ethan Fan \\
    Department of Biostatistics, Duke University, Durham, USA \\
   Tengfei Li and    Hongtu Zhu\thanks{ Address for correspondence: Hongtu Zhu, Ph.D., E-mail: htzhu@email.unc.edu. \\
  This paper was primarily completed during Dr. Ting Li's visit to UNC Chapel Hill. Dr. Fan's work  was partially supported by the National Science Foundation (NSF) grants DMS-2230795 and DMS-2230797. Dr. Zhu's work was partially supported by the Gillings Innovation Laboratory on generative AI and by grants from the National Institute on Aging (NIA) of the National Institutes of Health (NIH), including U01AG079847, 1R01AG085581, and RF1AG082938, and the NIH grants R01AR082684, and 1OT2OD038045-01. The content of this paper is solely the responsibility of the authors and does not necessarily represent the official views of these institutions.
  } \\
    Departments of Radiology, Computer Science, Genetics, and Biostatistics, \\ University of North Carolina at Chapel Hill, Chapel Hill, USA}
  \maketitle
} \fi

\if0\blind
{
  \bigskip
  \bigskip
  \bigskip
  \begin{center}
    {\LARGE \bf   Causal Inference in Biomedical Imaging with  Functional Linear Structural Equation Models}
\end{center}
  \medskip
} \fi

This supplementary material contains additional simulation results, additional real data analysis, auxiliary lemmas and proofs of the theorems in the main text.

\section{Additional Simulation Studies}
\label{sec:additional_simu}

In this section, we present additional simulation results to examine the finite sample performance of the proposed method in scenarios where the dimension of scalar covariates exceeds the sample size, and the proposed divide-and-conquer procedures. 

{\bf Example \ref{ex:one_dimensional} (Continued).}
The data settings are similar to those in Example \ref{ex:one_dimensional} in the main text, with the exception that the dimension \( p \) is set to 500. 

Tables \ref{table:MSEs_betaS_high} and \ref{table:MSEs_Cs_high} summarize the results, including variable selection accuracy, estimation accuracy for the outcome model, and the functional model. The findings are consistent with those observed under \( p = 20 \): the proposed method outperforms the PFLR method in the presence of endogeneity, while providing accurate estimates for both the outcome and exposure models.

	\begin{table}
		\begin{center}
			\caption{\small Simulation results of Monte Carlo averages and empirical standard errors in parentheses in Example \ref{ex:one_dimensional} for $p=500$.}
			\label{table:MSEs_betaS_high}
            \vspace{0.3ex}
		\tabcolsep 1pt
			\renewcommand{\arraystretch}{0.8}
			\begin{tabular}{ccccccccccccc}
				\hline
				\hline
				$\rho_1$&      $\rho_2$ &   & FZ$_Z$  &  FN$_Z$   &  FZ$_Y$  &  FN$_Y$  & MSE$_{\bm{B} }$  & MSE$_\beta$  \\
				\hline 
0.3	&	0	&	FLSEM	&	0.000(0.000) & 	1.193(0.686) & 	0.000(0.000) & 	1.875(0.896) & 	0.008(0.002) & 	0.158(0.065) & 	\\
	&		&	PFLM	&	- & 	-  & 	0.000(0.000) & 	3.913(0.428) & 	0.006(0.003) & 	0.194(0.036) & 	\\
\cline{2-9}												
	&	0.2	&	FLSEM	&	0.000(0.000) & 	1.250(0.702) & 	0.000(0.000) & 	1.866(0.904) & 	0.010(0.004) & 	0.184(0.076) & 	\\
	&		&	PFLM	&	- & 	-  & 	0.000(0.000) & 	3.894(0.444) & 	0.016(0.005) & 	0.189(0.039) & 	\\
\cline{2-9}												
	&	0.5	&	FLSEM	&	0.000(0.000) & 	1.257(0.677) & 	0.000(0.000) & 	2.041(0.905) & 	0.009(0.003) & 	0.248(0.094) & 	\\
	&		&	PFLM	&	- & 	-  & 	0.000(0.000) & 	3.942(0.405) & 	0.207(0.050) & 	0.470(0.111) & 	\\
\cline{2-9}												
	&	0.7	&	FLSEM	&	0.000(0.000) & 	1.243(0.686) & 	0.000(0.000) & 	2.094(0.966) & 	0.009(0.003) & 	0.275(0.107) & 	\\
	&		&	PFLM	&	- & 	-  & 	0.000(0.000) & 	3.923(0.421) & 	1.108(0.278) & 	8.955(3.187) & 	\\
\hline												
0.5	&	0	&	FLSEM	&	0.000(0.000) & 	0.476(0.497) & 	0.000(0.000) & 	1.777(1.130) & 	0.009(0.004) & 	0.153(0.079) & 	\\
	&		&	PFLM	&	- & 	-  & 	0.000(0.000) & 	3.856(0.508) & 	0.007(0.003) & 	0.202(0.039) & 	\\
\cline{2-9}												
	&	0.2	&	FLSEM	&	0.000(0.000) & 	0.449(0.495) & 	0.000(0.000) & 	1.790(1.037) & 	0.011(0.006) & 	0.192(0.091) & 	\\
	&		&	PFLM	&	- & 	-  & 	0.000(0.000) & 	3.836(0.537) & 	0.019(0.008) & 	0.201(0.040) & 	\\
\cline{2-9}												
	&	0.5	&	FLSEM	&	0.000(0.000) & 	0.471(0.497) & 	0.000(0.000) & 	1.909(0.918) & 	0.011(0.006) & 	0.249(0.101) & 	\\
	&		&	PFLM	&	- & 	-  & 	0.000(0.000) & 	3.875(0.477) & 	0.197(0.067) & 	0.493(0.194) & 	\\
\cline{2-9}												
	&	0.7	&	FLSEM	&	0.000(0.000) & 	0.439(0.494) & 	0.000(0.000) & 	1.925(0.943) & 	0.011(0.006) & 	0.281(0.112) & 	\\
	&		&	PFLM	&	- & 	-  & 	0.000(0.000) & 	3.913(0.428) & 	1.080(0.293) & 	8.595(3.449) & 	\\
    \hline \hline
			\end{tabular}
		\end{center}
	\end{table}

 	\begin{table}
		\begin{center}
			\caption{\small Simulation results for the functional response of Monte Carlo averages and empirical standard errors in parentheses in Example \ref{ex:one_dimensional} for $p=500$.}
			\label{table:MSEs_Cs_high}
			
			\vspace{2ex}
			\tabcolsep 4pt
			\renewcommand{\arraystretch}{0.8}
			\begin{tabular}{ccccccccccccc}
				\hline
				\hline
				$\rho_1$  &  $\rho_2$   & MSE$_{C_1}$ & MSE$_{C_2}$ & MSE$_{C_3}$ & MSE$_{C_4}$ & MSE$_{C_5}$ & PMSE$_{Z}$\\
				\hline 
0.3	&	0	&	0.004(0.003) & 	0.005(0.004) & 	0.006(0.004) & 	0.006(0.005) & 	0.006(0.005) & 	1.643(0.195) & 	\\
	&	0.2	&	0.005(0.003) & 	0.006(0.004) & 	0.005(0.004) & 	0.006(0.004) & 	0.005(0.004) & 	1.634(0.191) & 	\\
	&	0.5	&	0.006(0.004) & 	0.007(0.005) & 	0.006(0.004) & 	0.006(0.005) & 	0.006(0.005) & 	1.635(0.192) & 	\\
	&	0.7	&	0.005(0.003) & 	0.006(0.005) & 	0.006(0.004) & 	0.005(0.004) & 	0.006(0.004) & 	1.638(0.192) & 	\\
\hline										
0.5	&	0	&	0.007(0.004) & 	0.008(0.006) & 	0.007(0.006) & 	0.007(0.005) & 	0.007(0.005) & 	1.633(0.192) & 	\\
	&	0.2	&	0.006(0.004) & 	0.008(0.005) & 	0.007(0.006) & 	0.007(0.005) & 	0.007(0.005) & 	1.628(0.190) & 	\\
	&	0.5	&	0.005(0.003) & 	0.008(0.006) & 	0.009(0.006) & 	0.007(0.005) & 	0.007(0.006) & 	1.627(0.189) & 	\\
	&	0.7	&	0.006(0.004) & 	0.008(0.005) & 	0.008(0.006) & 	0.008(0.006) & 	0.007(0.005) & 	1.630(0.189) & 	\\
				\hline \hline  		
			\end{tabular}
		\end{center}
	\end{table}

{\bf Example \ref{ex:divide and conquer} (Continued).}
In this example, we evaluate the proposed divide‐and‐conquer procedures under the two‐dimensional design of Example~\ref{ex:2D}, with \(\rho_1=0.2\) and \(\rho_2=0.3\).

First, we split the \(n\) observations into two equal subsamples to estimate the outcome model. Figure~\ref{fig:divide_y_z} (top row) displays boxplots of the outcome‐prediction error and the MSEs for \(\mathbf{B}(t_1,t_2)\) and \(\beta\). Even with half the data, the divide‐and‐conquer estimator matches the accuracy and predictive performance of the full‐sample fit.

Next, we apply the region‐based divide‐and‐conquer scheme to estimate the functional coefficients \(C_1,\dots,C_5\). We tile the imaging domain into \(10\times10\) subregions and again split the sample into two subsets. Figure~\ref{fig:divide_y_z} (bottom row) shows boxplots of the MSE for each \(C_\ell(t_1,t_2)\). The subregion‐based estimates incur negligible efficiency loss compared to the full data, demonstrating that both divide‐and‐conquer strategies scale effectively without sacrificing accuracy.

     \begin{figure}
 	\centering
 	\includegraphics[height=4cm, width=5cm]{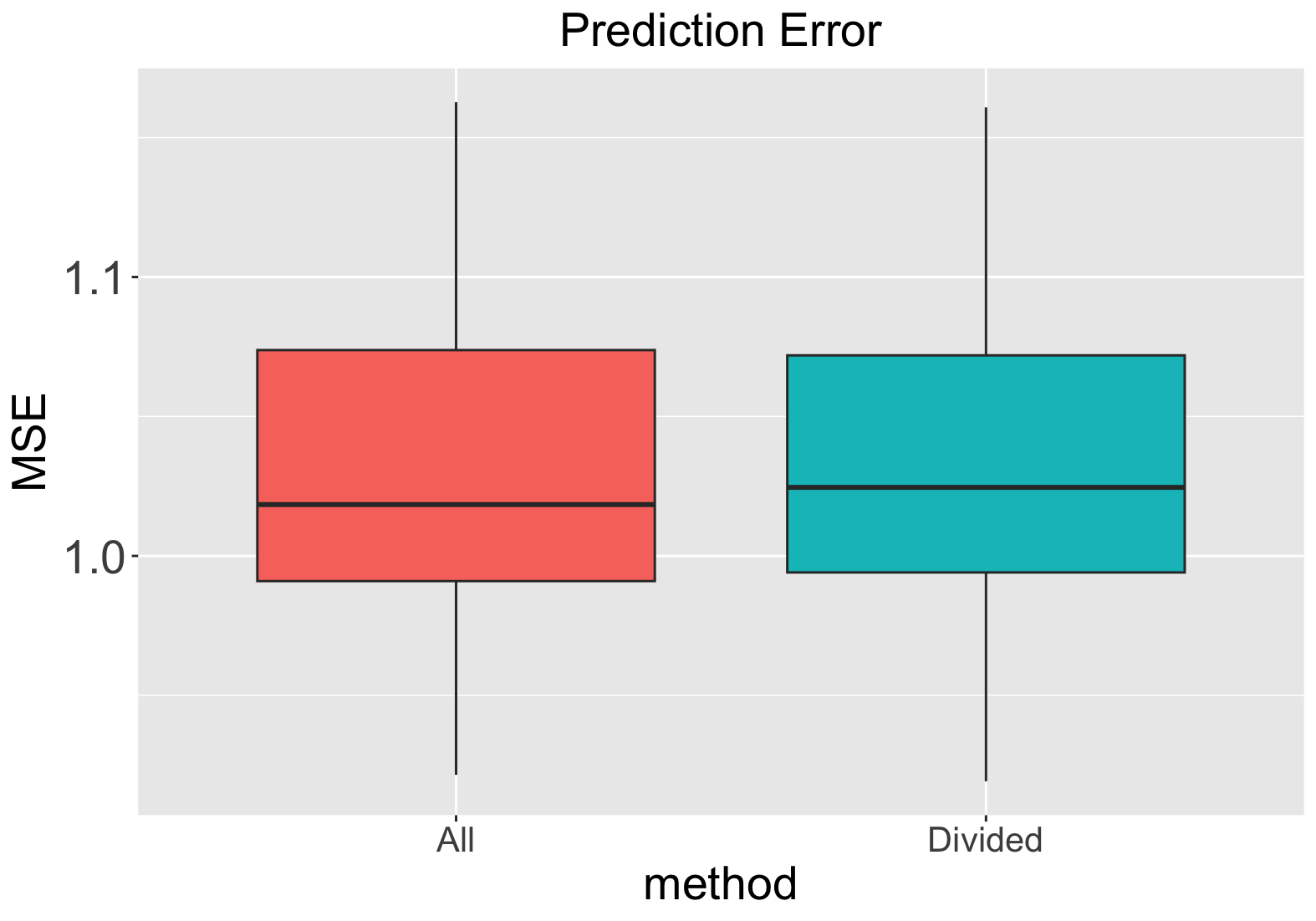}
       \includegraphics[height=4cm, width=5cm]{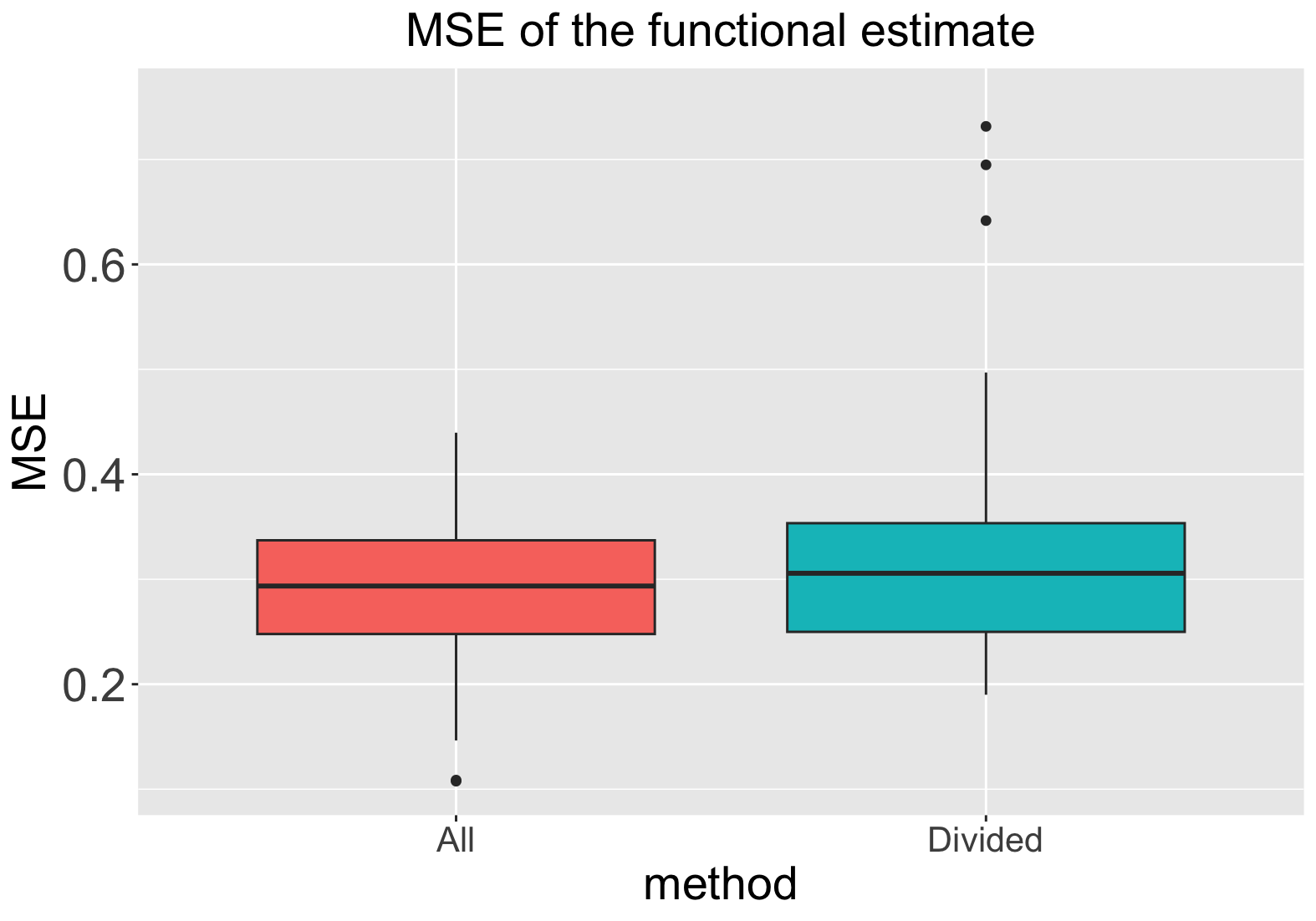}
        \includegraphics[height=4cm, width=5cm]{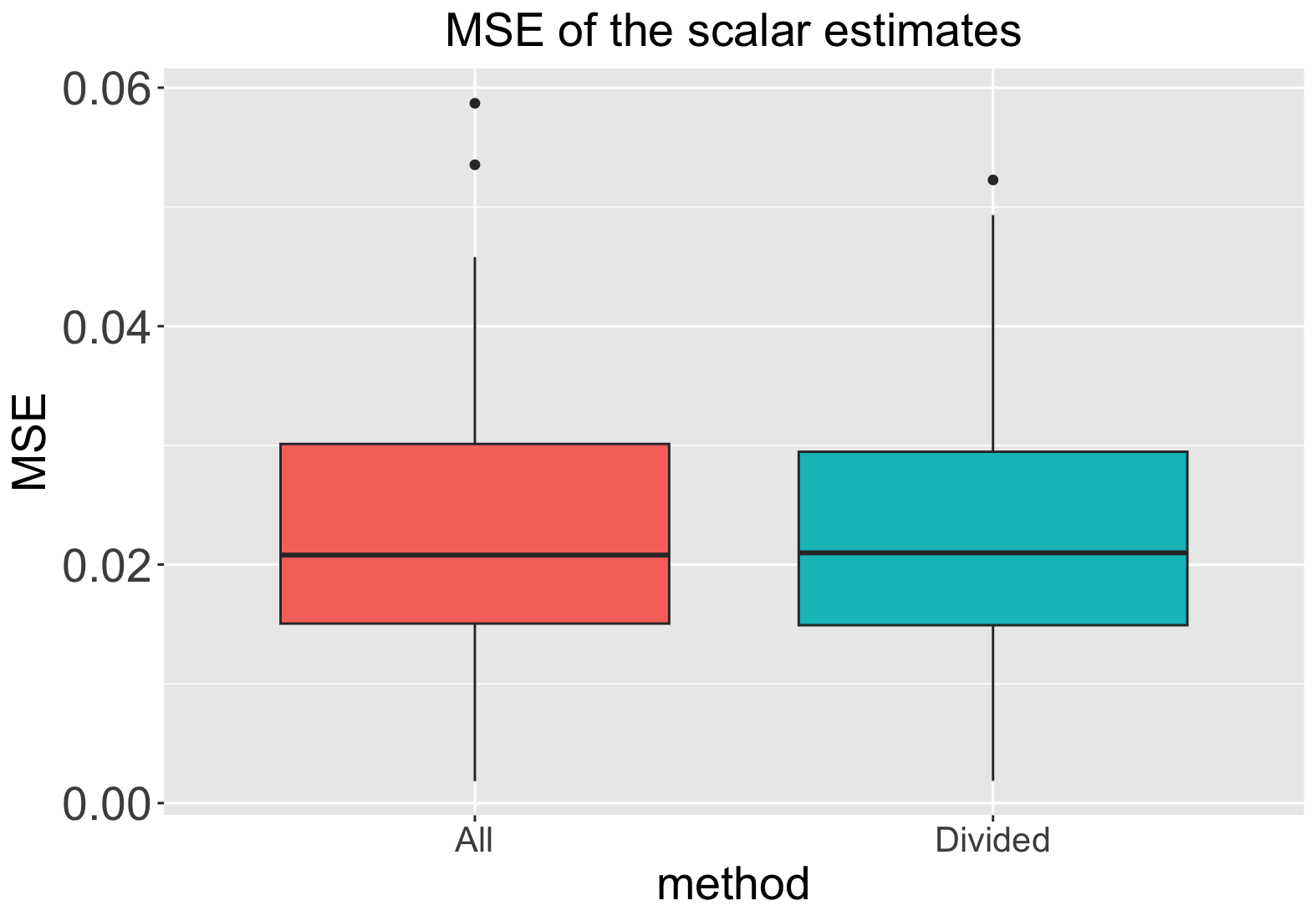} \\   	\includegraphics[height=4cm, width=5cm]{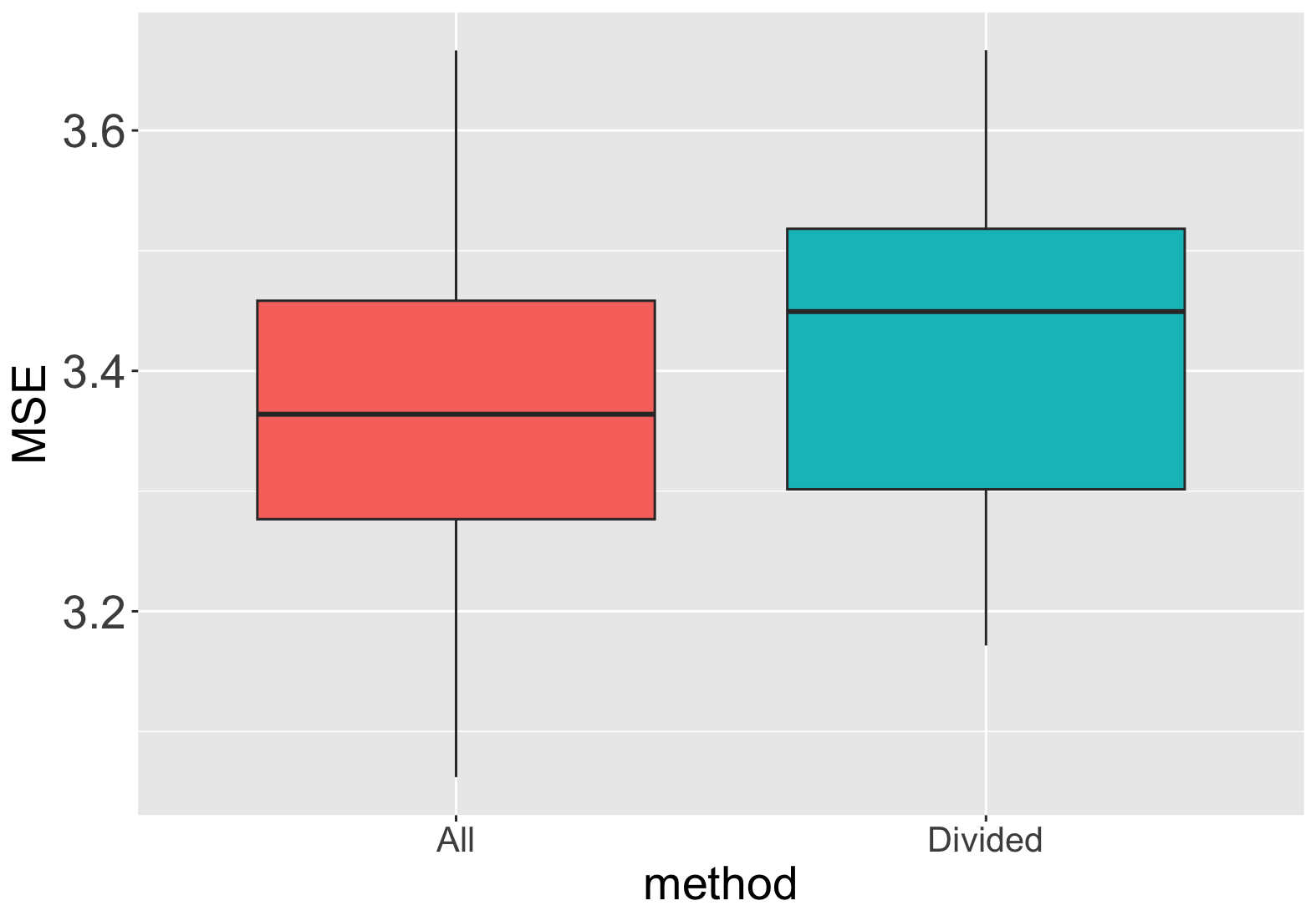}
 	\includegraphics[height=4cm, width=10cm]{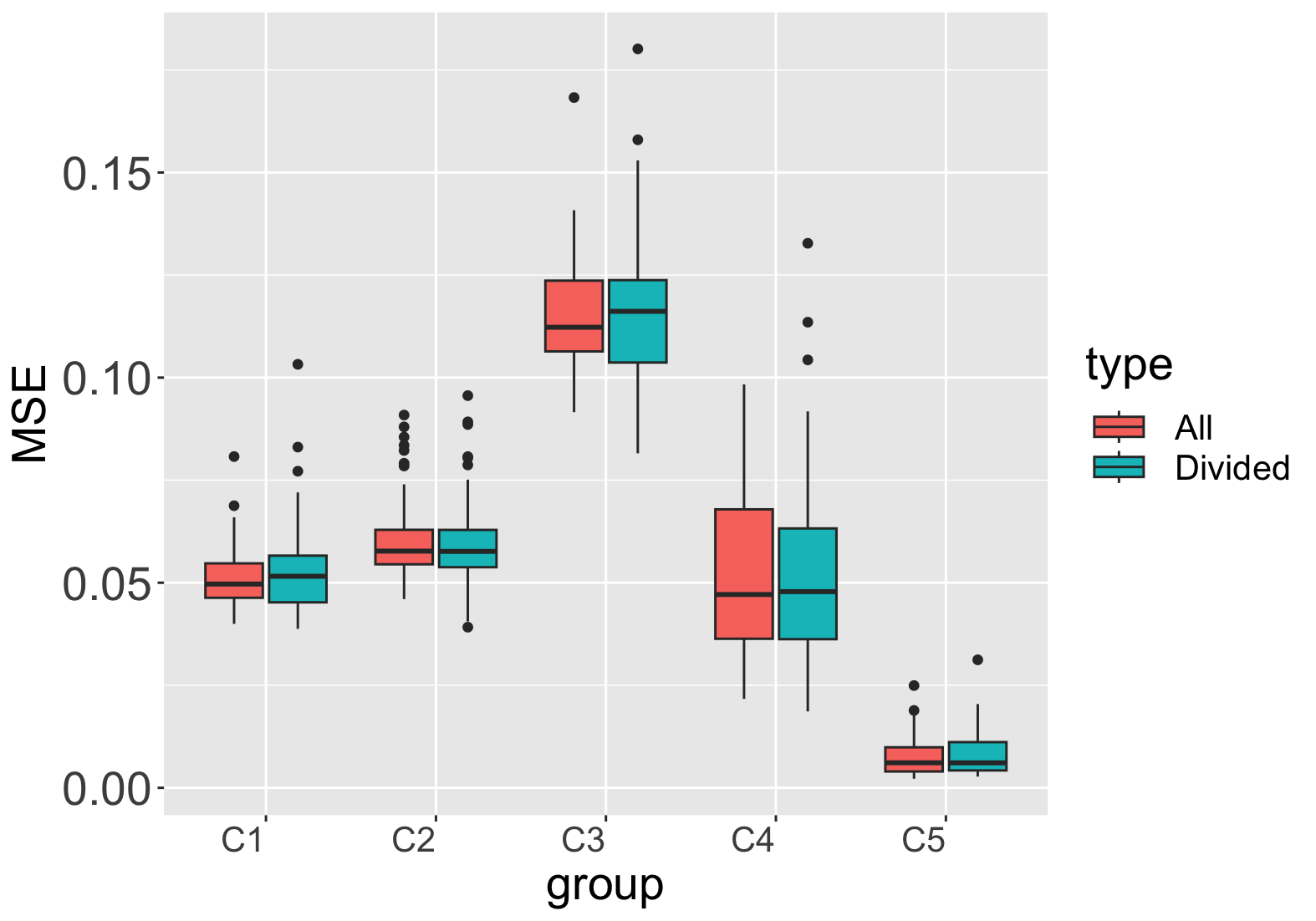}
 	\vspace{-0.1 in}
 	\caption{\small Boxplots of prediction errors and MSEs of Estimates for parameters in the outcome model (the first row), and for the functional model (the second row).  }
 	\label{fig:divide_y_z}
 	\vspace{-0.1 in}
 \end{figure}

\section{Additional Real Data Analysis}

Table \ref{table:common snp} presents the SNPs commonly selected in both the fluid intelligence and fMRI screening processes.

 \begin{table}[htbp]
 	\caption{Detailed information of the common SNPs  for fluid intelligence score and fMRI.}
 	\centering
 	\scriptsize
 	\tabcolsep 4pt
 	\renewcommand{\arraystretch}{0.8}
 	\begin{tabular}{lllrrr}
 		\hline \hline
 		\multicolumn{1}{c}{SNP}&\multicolumn{1}{c}{Chr}&\multicolumn{1}{c}{Base Pair} & Allele 1 & Allele 2  \\
 		\hline
 		rs76443071   & 15  &  32943439   & T  &  C    \\  
 		rs118151036   &  15   &  329454622 &  C  & A  \\ 
 		rs143826519   &  15  & 32947334  & T  & C  \\ 
 		\hline \hline
 	\end{tabular}
 	\label{table:common snp}
 \end{table}



\section{Proofs of the main theorems}

\begin{proof}[\bf Proof of Theorem \ref{thm:invalid_linear}]
Under the condition
\[
\int_{\mathcal{T}} C(t)\bm{B}(t)\,dt = \sum_{k=1}^R b_k \bm{c}_k + O(R^{-2r})
\]
for some \(r > 0\), the identification of the functional parameter \(\bm{B}(t)\) can be reduced to identifying the finite-dimensional coefficient vector \(\{b_k\}_{k=1}^R\), by noting that \(\bm{B}(t) \approx \sum_{k=1}^R b_k \varphi_k(t)\), where \(\{\varphi_k(t)\}\) denotes an orthogonal basis. 


We establish that \(\beta\) is uniquely identified if and only if \(\{b_k\}_{k=1}^R\) is uniquely determined by the moment conditions in equation~\eqref{eq:decomposition_identification}. Suppose \(\{b_k\}_{k=1}^R\) is uniquely identified. Then, for any two solutions \(\beta^{(1)}\) and \(\beta^{(2)}\), we have:
\[
\Gamma^* = \beta^{(1)} + \sum_{k=1}^R b_k \bm{c}_k
\quad \text{and} \quad
\Gamma^* = \beta^{(2)} + \sum_{k=1}^R b_k \bm{c}_k.
\]
Subtracting the two equations yields \(\beta^{(1)} = \beta^{(2)} = \Gamma^* - \sum_{k=1}^R b_k \bm{c}_k\), confirming the uniqueness of \(\beta\). Therefore, it suffices to show that \(\{b_k\}_{k=1}^R\) is uniquely identifiable.

First, we prove sufficiency. Assume there exist two solutions 
\((\beta^{(1)}, \{b_k^{(1)}\})\) and \((\beta^{(2)}, \{b_k^{(2)}\})\) to \eqref{eq:decomposition_identification} such that
\[
\Gamma^* = \beta^{(1)} + \sum_{k=1}^R b_k^{(1)} \,\bm c_k
\quad\text{and}\quad
\Gamma^* = \beta^{(2)} + \sum_{k=1}^R b_k^{(2)} \,\bm c_k.
\]
Let 
$
\mathcal C^{(d)} = \{\,j : \beta_j^{(d)} \neq 0\}
$
denote the set of invalid instruments under \(\beta^{(d)}\) for $d=1$ and $2$. Since \(|\mathcal C^{(d)}|<U\), its complement satisfies
\[
\bigl|(\mathcal C^{(d)})^C\bigr| = L - |\mathcal C^{(d)}| > L - U.
\]
Choose subsets 
$
S_1 \subset (\mathcal C^{(1)})^C$ and  $S_2 \subset (\mathcal C^{(2)})^C
$
of size \(|S_1| = |S_2| = L - U + 1 > K\).  For each \(j\in S_1\) and \(\nu\in S_2\), the moment equations reduce to
\[
\Gamma^*_j = \sum_{k=1}^R b_k^{(1)}\,c_{jk}\quad\text{and}\quad 
\qquad
\Gamma^*_\nu = \sum_{k=1}^R b_k^{(2)}\,c_{\nu k}.
\]
By the subspace‐restriction condition, it follows that \(b_k^{(1)} = b_k^{(2)}\) for all \(k\).  Consequently,
\[
\beta^{(1)} 
= \Gamma^* - \sum_{k=1}^R b_k^{(1)}\,\bm c_k
= \Gamma^* - \sum_{k=1}^R b_k^{(2)}\,\bm c_k
= \beta^{(2)},
\]
contradicting the assumption of distinct solutions. Hence, the solution is unique.


Second, we prove necessity by showing that the subspace-restriction condition on \(\Gamma^*\) and \(\{\bm c_k\}\) must hold if the solution is unique.  
Such a set always exists by taking a subset of the true valid instruments. Specifically, let \(\mathcal C^*\) denote the true set of invalid instruments.  
Since \(|(\mathcal C^*)^C| > L - U\), we can select a subset \((\mathcal C^{*\prime})^C \subseteq (\mathcal C^*)^C\) of size \(L - U + 1\) such that for each \(j \in (\mathcal C^{*\prime})^C\),
\[
\Gamma_j^* = \sum_{k=1}^R b_k^{(*\prime)} c_{jk}.
\]
If this subset \((\mathcal C^{*\prime})^C\) is unique, the subspace-restriction condition holds trivially since \(\{b_k^{(*\prime)}\}\) must match the true coefficients \(\{b_k\}\).

Now assume there are two or more subsets. Let \(\mathcal C^{(1)}\) and \(\mathcal C^{(2)}\) be any two subsets of \(\{1,\dots,L\}\) whose complements each have size 
\(\bigl|(\mathcal C^{(d)})^C\bigr|=L-U+1>R\) for $d=1,2$.  Denote by \(\{b_k^{(d)}\}\) the coefficients inferred from \((\mathcal C^{(d)})^C\) for $d=1,2$ satisfying $\Gamma_j^* \;=\;\sum_{k=1}^R b^{(d)}_k\,c_{jk}$ for $j \in (\mathcal C^{(d)})^C$. We aim to show that $b_k^{(1)} = b_k^{(1)}$.

For each \(d=1,2\), we construct two solution sets $(\beta^{(d)}, \{ b_k^{(d)} \} ) $ to \eqref{eq:decomposition_identification} for $d=1,2$ by taking $\beta_j^{(d)}=0$ for $j\in(\mathcal C^{(d)})^C$ and $\beta_j^{(d)} = \Gamma_j^* -\sum_{k=1}^R b_k^{(d)}\,c_{jk}$ for $j\in\mathcal C^{(d)}$. It is obvious that the two solution sets satisfy 
\[
\begin{cases}
\Gamma_j^* = \displaystyle\sum_{k=1}^R b_k^{(d)}\,c_{jk}, & j\in(\mathcal C^{(d)})^C,\\[1ex]
\Gamma_j^* = \beta_j^{(d)}+\sum_{k=1}^R b_k^{(d)}\,c_{jk}, & j\in\mathcal C^{(d)},
\end{cases}
\]
where \(\beta^{(d)}\) has fewer than \(U\) nonzero entries. 
Uniqueness of the overall solution implies \(b_k^{(1)}=b_k^{(2)}\).  Since the choice of \(\mathcal C^{(1)}\) and \(\mathcal C^{(2)}\) was arbitrary, we conclude that \(b_k^{(m)}=b_k^{(m')}\) for all \(m,m'\), establishing necessity.



\end{proof}

\begin{proof}[\bf Proof of Corollary \ref{cor:number of invalid IV}.]
Consider any two sets \( S_m \) and \( S_{m'} \) with corresponding coefficients \( \{ b_k^{(m)} \} \) and \( \{ b_k^{(m')} \} \) satisfying 
\[
\sum_{k=1}^R b_k^{(m)} \bm{c}_{\ell k} = \Gamma_\ell^*, \quad \ell \in S_m,  \quad \text{and} \quad
\sum_{k=1}^R b_k^{(m')} \bm{c}_{\ell k} = \Gamma_\ell^*, \quad \ell \in S_{m'}.
\]
We now show that \( S_m \cap S_{m'} \) contains at least \( R \) elements. Since \( |S_m| = |S_{m'}| = L - U + 1 \geq (L + R - 1)/2 + 1 \), we have
\[
|S_m| + |S_{m'}| \geq L + R + 1 
\]
implying \( |S_m \cap S_{m'}| \geq R \). For each \( \ell \in S_m \cap S_{m'} \), the equality of representations implies \( b_k^{(m)} = b_k^{(m')} \). Since this holds for any such pair \( S_m \) and \( S_{m'} \), the subspace-restriction condition in Theorem~\ref{thm:invalid_linear} must hold whenever \( U \leq (L - R + 1)/2 \), yielding identification as desired.
\end{proof}

\begin{proof}[\bf Proof of Theorem \ref{thm:imaging_approximation_error}]
	Suppose that $\gamma \leq 1$. By iteratively using \eqref{eq:D2 iteration} in Lemma \ref{eq: D2 decomposition},
	we have 
	\bas
	D_2 ( A_R^{k+1} ) 
	& \leq &
	\gamma D_2(  A_R^k ) + \gamma h(J, \lambda_K) \\
	& \leq &
	\gamma ( \gamma D_2(  A_R^{k-1} ) + \gamma h(J, \lambda_K) ) + \gamma h(J, \lambda_K)  \\
	& \leq &  
	\cdots  \\
	& \leq &
	\gamma^{k+1} D_2 (A_R^0 ) + {\gamma \over (1 - \gamma) }h(J, \lambda_K)  \\
	& \leq &
	\gamma^{k+1} \|   \bm{C}^*  \|_2  + {\gamma \over (1 - \gamma) } h(J, \lambda_K) .
	\eas 
	Further, plugging the above inequality into\eqref{eq: beta_iteration},
	it is easy to derive that 
	\bas
	\|  \bm{C}^{k+1} - \bm{C}^*  \|_2 
	& \leq & 
	\Big( 1+ { \theta_{J,J}  tr(K) \over c_-(J) tr(K + \lambda_K I) } \Big)  D_2(A^k)+ 
	 h(J,\lambda_K)   \\
	& \leq & 
\Big( 1+ { \theta_{J,J}  tr(K) \over c_-(J) tr(K + \lambda_K I) } \Big) \big(  \gamma^{k} \|  \bm{C}^*  \|_2  \big) \\
	&&+
	\Big[
	{\gamma \over (1 - \gamma) }  \Big( 1+ { \theta_{J,J}  tr(K) \over c_-(J) tr(K + \lambda_K I) } \Big)  + 1
	\Big] h(J,\lambda_K).
	\eas
 Hence,
 \begin{eqnarray}
 \label{eq:b_1_b_2}
  b_1 = 1+ { \theta_{J,J}  tr(K) \over c_-(J) tr(K + \lambda_K I)} , \quad
 b_2 = {\gamma \over (1 - \gamma) }  \Big( 1+ { \theta_{J,J}  tr(K) \over c_-(J) tr(K + \lambda_K I) } \Big)  +1 .    
 \end{eqnarray}
    The order of $h(J, \lambda_K)$ follows from Lemma \ref{lem: KZ_epsilon concentration}.
	This completes the proof.
\end{proof}

\begin{proof}[\bf Proof of Theorem \ref{thm:consistency_revelant_controls}]
  
  This proof follows directly from the proof of Theorem 2 and Theorem 3 in the supplementary of 
  \cite{li2024partially} by replacing the functional covariate with the estimated functional covariate, and using Lemma \ref{lem: Tn_T},
\begin{eqnarray}
\label{eq:M_3_M_4}
M_3 = c_1 M^2_4 \delta^2_{\max} \| \bm{B}^*  \|^2_\mathcal{K} + \sigma_x^2 n^{-1/2} M_4 \|\bm{B}^*  \|^2_\mathcal{K} ,  \quad
M_4 =   n^{1/2} \lambda ( 2M_1  +1) +  ( 1 + {1 \over \sqrt{ \nu_4}} ) \sqrt{c_1 } tr(T),   
\end{eqnarray}
  and we omit it here.

\end{proof}

\begin{proof}[\bf Proof of Theorem \ref{thm:testing}]
Recall that $\mathcal{M}_{\widehat A} = I_n - X_{\widehat A} ( X_{\widehat A} X_{\widehat A}^\top )^{-1} X_{\widehat A}^\top$.
According to the outcome model,
$$
\mathcal{M}_{\widehat A}  Y = \mathcal{M}_{\widehat A}  X_{A^*} \beta_{A^*}^* + \mathcal{M}_{\widehat A} \int_{\mathcal{T}} \widehat{\bm{Z}} (t) \bm{B}^*(t) dt + \mathcal{M}_{\widehat A} \widehat \nu,
$$
where $\widehat \nu = \int_{\mathcal{T}} (\bm{Z} (t)- \widehat{\bm{Z}}(t) )\bm{B}^*(t) dt + \epsilon. $
 By the definition of $\widehat{\bm{F}} (t)$, under the null hypothesis that $\bm{B}^*(t)=0$, we can obtain that 
 \begin{eqnarray*}
   \sqrt{n} (T_n + \lambda I)^{-1} \widehat{\bm{F}} (t)
 &=&
 n^{-1/2}  ( K^{1/2} \widehat{\bm{Z}}(t) )^\top \mathcal{M}_{\widehat A}  Y 
 =
 n^{-1/2}  ( K^{1/2} \widehat{\bm{Z}}(t) )^\top  ( \mathcal{M}_{\widehat A}  X_{A^*} \beta_{A^*}^*  + \mathcal{M}_{\widehat A} \epsilon ) \\
&=&
  n^{-1/2}  ( K^{1/2} \widehat{\bm{Z}}(t) )^\top  \Big (  ( \mathcal{M}_{\widehat A} - \mathcal{M}_{ A^*} )\epsilon +( \mathcal{M}_{\widehat A} - \mathcal{M}_{ A^*} ) X_{A^*} \beta_{A^*}^*  +  \mathcal{M}_{A^*} \epsilon 
  \Big) \\
  &=&
  S_1 +S_2+S_3
 \end{eqnarray*}
 where $S_1, S_2, S_3$ denote the above three terms, and 
 the third equality follows from that $\mathcal{M}_{ A^*}  X_{A^*} =0 $.

 We first calculate $\|S_1\|_2^2$. For any $\nu >0$, 
 \begin{eqnarray*}
  P( \|S_1\|_2^2 > \nu )  \leq P( \|S_1\|_2^2 > \nu, \widehat A=A^* )  + P(\widehat A \neq A^* )=0
 \end{eqnarray*}
 as $n \rightarrow \infty$.
 The inequality follows from that when $\widehat A=A^*$, $\|S_1\|_2^2=0$ due to 
$\mathcal{M}_{\widehat A} = \mathcal{M}_{ A^*}$. 
Hence $\|S_1\|_2^2=O_p(1)$.
Similarly, $\|S_s\|_2^2=O_p(1)$.

We first analyze its finite-dimensional distributions of $S_3$. For any set of points $t_1, t_2, \dots t_m$, consider the vector
$$
n^{-1/2} \big(\epsilon^\top  \mathcal{M}_{A^*}  K^{1/2} \widehat{\bm{Z}}(t_1)  ,\dots, 
\epsilon^\top  \mathcal{M}_{A^*} K^{1/2} \widehat{\bm{Z}}(t_m) \big)^\top
= H \epsilon,
$$
where $H= n^{-1/2} (\mathcal{M}_{A^*}  K^{1/2} \widehat{\bm{Z}}(t_1)  ,\dots, 
 \mathcal{M}_{A^*}  K^{1/2} \widehat{\bm{Z}}(t_m)  )^\top \in \mathbb{R}^{k \times n}$.
Since $\epsilon_i$ are independent with mean 0 and variance $\sigma^2$, the vector $G \epsilon$ has mean 0 and covariance matrix 
$Cov(H \epsilon)=\sigma^2 G G^\top$.
Notice that by the law of large numbers,
$n^{-1} H H^\top$ converges  in probability to the matrix $\Sigma$
$$
\Sigma_{ij} = n^{-1} K^{1/2}E \big( \tilde{\bm{Z}}^\top (t_i) \mathcal{M}_{A^*}  \tilde{\bm{Z}} (t_j) \big) K^{1/2} \sigma^2
$$
by noticing that $|E \big( \widetilde{\bm{Z}}^\top (t_i) \mathcal{M}_{A^*}  \widetilde{\bm{Z}} (t_j) \big) -E \big( \widehat{\bm{Z}}^\top (t_i) \mathcal{M}_{A^*}  \widehat{\bm{Z}} (t_j) \big) |=o(1)$
which is positive definite.
According to the multivariate central limit theorem, we can obtain that 
\begin{eqnarray*}
\label{eq:finite dimensional}
   H \epsilon \overset{d}{\rightarrow} N(0, \Sigma).  
\end{eqnarray*}

Next, we show the tightness of $S_3$.
Consider the term $\tilde S_3=n^{-1/2} K^{1/2} \widetilde{\bm{Z}}(t)   \epsilon = n^{-1/2} \sum_{i=1}^n  K^{1/2} \widetilde{\bm{Z}}_i (t)   \epsilon_i $.
Since $$ 
E \|  K^{1/2} \widetilde{\bm{Z}}_i (t)   \epsilon_i\|_2^2 = \sigma^2 \int_{\mathcal{T}} \int_{\mathcal{T}} C(s,t) K(s,t) ds dt < \infty,
$$
and $\{ K^{1/2} \widetilde{\bm{Z}}_i (t)   \epsilon_i \}$ are independent and identically distributed random elements, then $\tilde S_3$ converges in  distribution to a Gausasian process according to Theorem 7.7.6 of \cite{hsing2015theoretical}.
It follows directly that $\tilde S_3$ is tight.
Hence, $S_3$ is also tight by noticing that 
$ \sup_t|  K^{1/2} \widetilde{\bm{Z}}(t) \mathcal{M}_{A^*}  \epsilon | <  \sup_t|  K^{1/2} \widetilde{\bm{Z}}(t)   \epsilon | $.

According to limit of the finite-dimensional distributions in \eqref{eq:finite dimensional} and the tightness of $S_3$, Theorem 18.14 in \cite{van2000asymptotic} indicates that $S_3$ converges in distribution to a Gaussian process with the covariance 
$$
 S_3 \overset{d}{\rightarrow} GP \big(0, n^{-1} K^{1/2}E \big( \widetilde{\bm{Z}}^\top (s) \mathcal{M}_{A^*}  \widetilde{\bm{Z}} (t) \big) K^{1/2} \sigma^2 \big).
$$

By the Karhunen-Loeve expansion, the above Gaussian process $G$ admits the decomposition 
$G = \sigma\sum_{j=1}^\infty \sqrt{\tilde s_j} \xi_j \psi_j $, where $\{ \xi_j, j \geq 1 \}$ are independent and identically distributed standard Gaussian random variables, ${\psi_j, j \geq 1}$ are orthogonal basis, and $\{\tilde s_j\}$ are eigen values of the operator $ n^{-1/2} K^{1/2}E \big( \widetilde{\bm{Z}}^\top (s) \mathcal{M}_{A^*}  \widetilde{\bm{Z}} (t) \big) K^{1/2}$. Hence, 
$$
{ \|S_3 \|_2^2 \over \sigma^2 } \overset{d}{\rightarrow}\sum_{j=1}^\infty \tilde s_j \xi_j^2.
$$
Combining the above results with $\|S_1\|_2^2=O_p(1) , \|S_2\|_2^2=O_p(1)$ and $\widehat \sigma^2$ converge in probability to $\sigma^2$ by large law of numbers,
we have 
$$
S_n \overset{d}{\rightarrow}\sum_{j=1}^\infty \tilde s_j \xi_j^2.
$$
This completes the proof.
\end{proof}

\section{Auxiliary Lemmas}
 We first define some notations to facilitate our discussions. 
Recall that $A_R^*$ is the true index set of the nonzero  variables, $\bm{C}^*$ is the true value of the scalar coefficients in the exposure model,
and 
$\{ A_R^k \}_k$ is the sequence of active sets generated by the FSDAR algorithm.
For any given integers $J$ and $J_z^*$ with $J \geq J_z^*$ and $F \subset S$ with $|F|= J- J_z^*$,
let $A^o= A_R^* \bigcup F$ and $I^o = (A^o)^c$. 
Define
\begin{eqnarray*}
\label{eq: D2 and triangle definition}
D_2 ( A_R^k ) = \|  \bm{C}^* |_{A_R^* \backslash A_R^{k} }    \|_2,  \quad \bigtriangleup^k = \bm{C}^{k+1} |_{A_R^k}- \bm{C}^*|_{A_R^k}.
\end{eqnarray*}
The term $D_2 ( A_R^k )$ measures the $L_2$ norm of the false zero coefficients in the $k$-th iteration, namely the coefficients in $A_R^*$ but not in $A_R^{k}$.
The term $\bigtriangleup^k$ measures the bias of the estimated coefficients in $A_R^{k}$.

Then, we let
$$
\begin{aligned}
A_1^k = {A_R^k \cap  A^o},  A_2^k = {A^o \backslash A_1^k},  I_3^k = {A_R^k \cap I^o}, I_4^k ={  I^o \backslash I_3^k },\\
A_{11}^k = { A_1^k   \backslash ( { A^{k+1} \cap A_1^k} ) },A_{22}^k = { A_2^k \backslash ( {A^{k+1} \cap A_2^k} )},
I_{33}^k =  A^{k+1} \cap I_3^k, I_{44}^k =  A^{k+1} \cap I_4^k.
\end{aligned}
$$
At the $k$-th iteration, the sets $A_1^k$ and $I_3^k$ contain the true and false positives in the active set $A_R^k$, respectively, while $A_2^k$ and $I_4^k$ contain the false and true negatives. The sets $A_{11}^k$ and $A_{22}^k$ include indices in $A^o$ that will be removed in the next iteration, and $I_{33}^k$ and $I_{44}^k$ contain indices in $I^o$ that will be added in the next iteration.
Because the sparsity level in the algorithm is $J$, then
$ | A_R^{k} | =| A_R^{k+1}|=J$. Denote the cardinality of $I_3^k$ by $\ell_k = |  I_3^k |$. 
It is obvious that $A_R^k = A_1^k \bigcup I_3^k$,
$ A_2^k = |A^o|- |A_1^k| =  |A^o| - ( | A_R^k | -  |  I_3^k |)= J-( J-\ell_k )= \ell_k$.
By definition,
\bas
| A_{11}^k | + |  A_{22}^k |= |  I_{33}^k |+|  I_{44}^k |, \quad  D_2 ( A^k ) = \|   \bm{C}^* |_{A^o \backslash A_R^{k} }  \|_2 =  \|   \bm{C}^* |_{A_2^k }  \|_2.
\eas

\begin{lemma}
	\label{lem:eigen value}
	Letting $A$ be the subset of $\mathcal{S}$ with $|A|=J$,
	suppose that Assumptions \ref{assum: SRC} holds,  then
	\begin{eqnarray}
		\label{eq:design_matrix_norm}
	0 <	n c_-(J) tr( K + \lambda_K I)  \leq \|  (\bm{X}_A^\top  \bm{X}_A) \otimes (K^\top_{\bm{u}} K_{\bm{u}} )+ \lambda I \|_{op} \leq  n c_+(J) tr( K + \lambda_K I) <  \infty, \\
		\label{eq:design_matrix_inverse_norm}
	0 <	{tr^{-1}( K + \lambda_K I) \over n  c_- (J) }  \leq \| \big[ (\bm{X}_A^\top  \bm{X}_A) \otimes (K^\top_{\bm{u}} K_{\bm{u}} )+ \lambda I \big]^{-1} \|_{op} \leq  {tr^{-1}( K + \lambda_K I) \over n c_+ (J) }   <  \infty
	\end{eqnarray}
\end{lemma}
\begin{proof}
We first prove \eqref{eq:design_matrix_norm}. According to Lemma 7.2 of \cite{alizadeh1998primal}, the eigenvalue of the Kronecker product of two matrices is the product of the corresponding eigenvalues.
By definition, 
\begin{eqnarray*}
&&	\|  (\bm{X}_A^\top  \bm{X}_A) \otimes (K^\top_{\bm{u}} K_{\bm{u}} )+ \lambda I   \|_{op} 
	=
	\sup_{h: \| h \|_{L_2} =1} | \langle  [(\bm{X}_A^\top  \bm{X}_A) \otimes (K^\top_{\bm{u}} K_{\bm{u}} )+ \lambda I] h, h \rangle    |  \\
	&&\quad \leq  
	n c_+ (J) \sup_{h: \| h \|_{L_2} =1}       \Big|   \int ( K(u,s) + \lambda I) h(s) ds    \Big|  \leq n c_+ (J)  tr(K  + \lambda I).
\end{eqnarray*}

Then it is easy to check that
\begin{eqnarray*}
n c_-(J) tr( K + \lambda I )  \leq \|  (\bm{X}_A^\top  \bm{X}_A) \otimes K^\top_{\bm{u}} K_{\bm{u}} + \lambda  \|_{op} \leq n c_+(J) tr( K + \lambda I)
\end{eqnarray*}
The derivation of \eqref{eq:design_matrix_inverse_norm} follows directly.
\end{proof}

\begin{lemma}
	\label{lem: KZ_epsilon concentration}	
	Suppose that Assumption \ref{assump:error of imaging} holds,
	then for any $\nu \in (0,1/2)$, 
		{\footnotesize
\begin{eqnarray}
	\label{eq: concentration of FLR_1}
P \Big( 
h(J, \lambda_K)
\leq\sigma_E  \sqrt{ J \log \left( {p \over \nu } \right)}  \sqrt{ {tr( K (K + \lambda I)^{-2} )  \over n m }+ n^{-1}  } +  J \sqrt{ \lambda_K/2}  
\Big) \geq 1 - \nu,
	\end{eqnarray}  }
 where $h(J, \lambda_K) =\max_{|A|= J}  \| (T_{nm}^{A} + \lambda_K I)^{-1} ( (nm)^{-1}\sum_i \sum_j E_i(u_j)  K_{u_{ij}} X_{i, A} + \lambda_K \bm{C}^*_{A} )\|_2$ with $T_{nm}^{A}$ being an opeartor such that
$T_{nm}^{A}(f)= (nm)^{-1} \sum_i \sum_j \langle X_{i,A}^\top f, K_{u_{ij}} \rangle_{\mathcal{K}}    K_{u_{ij}} X_{i, A}^\top $.
\end{lemma}

\begin{proof}
Recall that $X_j=(X_{1j}, \dots, X_{nj})^\top$ , denote
\begin{eqnarray*}
h(J) &=&\max_{ |A| \leq T }  \|  ( (X_A^\top X_A) \otimes K_{\bm{u}}^\top K_{\bm{u}}  + \lambda I )^{-1} {1 \over n m} \sum_{i=1}^n \sum_{j=1}^m K_{u_j} X_{ij} E_i(u_j) \|_{L_2} \\
& \leq &
\sqrt{|A|} \max_{j \in \{1, \dots, p\}}  \|  (X_j^ \top X_j ) \otimes K_{\bm{u}}^\top K_{\bm{u}}  + \lambda I )^{-1} {1 \over n m} \sum_{i=1}^n \sum_{j=1}^m K_{u_j} X_{ij} E_i(u_j) \|_{L_2} \\
&:=&
\sqrt{|A|} \max_{j \in \{1, \dots, p\}}  {1 \over nm} G_j. 
\end{eqnarray*}

Due to the fact $X_j^ \top X_j=1$, 
direct calculations lead to $\Mean( G_j )=0$ and 
\begin{eqnarray}
\label{eq:C_eset_variance}
	\Var( G_j ) &=& 
	\Mean \Big[ 
	( K + \lambda_K I )^{-1} {1 \over n m} \sum_{i=1}^n \sum_{j=1}^m K_{u_j} X_{ij} E_i(u_j) 
	\Big]^2 
	=
\sigma_E^2 \Big( {tr( K^2 (K + \lambda_K I)^{-2} )  \over n m }+ n^{-1}  \Big)
\end{eqnarray}
following a similar spirit to (S6.8) in \cite{li2025semi}.
As pointed out in \cite{zhang2005learning}, $tr( K^2 (K + \lambda_K I)^{-2} ) \leq tr( K (K + \lambda_K I)^{-1} )$ and the right right hand side term is simpler.
	Since $E_i(u_j)$ is sub-Gaussian, and by the union bound, for any $\nu \in (0,1/2)$, we have
$$
\prob( \max_j {1 \over nm } G_j > t ) \leq 2 p \exp( - { n^2 m^2 t^2 \over  2 \Var(G_j) } ) \leq \nu,
$$
where we take $t$ as
$$
t = \sqrt{ 4 \log(2 p /\nu) \Big[ \sigma_E^2 \Big( (nm)^{-1} {tr( K (K + \lambda_K I)^{-1} )   }+ n^{-1}  \Big)  \Big]  }.
$$
For the bias term, denote $\{ \kappa_k\}$ and $\{\phi_k(t)\}$ as the eigenvalues and eigenfunctions of the kernel $K$. For any $\ell$, $\bm{C}^*_{\ell}$ admits the expansion $\bm{C}^*_{\ell}(t) = \sum_{k=1}^\infty \bm{C}^*_{\ell k} \phi_k(t)$ and 
$\sum_{k=1}^\infty {\bm{C}^*_{\ell k}}^2/ \kappa_k <\infty$ as $\bm{C}^*_{\ell}$ resides in the RKHS.
It is easy to deduce that
\begin{eqnarray}
\label{eq:C_eset_bias}
   \| \lambda_K (K + \lambda_K I)^{-1} \bm{C}^*_{\ell} \|^2_2
=
\lambda_K  \sum_{k=1}^\infty { \lambda_K  \kappa_k \over (\lambda_K + \kappa_k)^2}  { {\bm{C}^*_{\ell k} }^2 \over \kappa_k }
\leq 
{\lambda_K  \over 2} \sum_{k=1}^\infty {{\bm{C}^*_{\ell k} }^2 \over \kappa_k }
=
{\lambda_K  \over 2}  \|\bm{C}^*_{\ell}\|^2_{\mathcal{K}}.
\end{eqnarray}

Recall that $|A|=J$, combining \eqref{eq:C_eset_variance} and \eqref{eq:C_eset_bias}
we can derive that 
\begin{eqnarray*}
h(J, \lambda_K) &=&\max_{|A|= J}  \| (T_{nm}^{A} + \lambda_K I)^{-1} ( (nm)^{-1}\sum_i \sum_j E_i(u_j)  K_{u_{ij}} X_{i, A} + \lambda_K \bm{C}^*_{A} )\|_2  \\
& \leq &
h(J) + J \max_{\ell} \| \lambda_K (K + \lambda_K I)^{-1} \bm{C}^*_{\ell} \|_2 \\
& \leq &
\sqrt{J \log ({p /\nu}) \sigma_E^2 \Big( (nm)^{-1} {tr( K (K + \lambda_K I)^{-1} )   } + n^{-1}  \Big)} + J \sqrt{ \lambda_K/2}. 
\end{eqnarray*}
This completes the proof.
\end{proof}

\begin{lemma}
\label{lem:minimizer}
	Let $C^\diamond$ be a coordinate-wise minimizer of the loss function with group L$_0$ penalty. Then, $C^\diamond$ satisfies that
\begin{eqnarray}
\label{eq:minimizer}
   			 d_\ell^\diamond = P^{-1}_{\ell}  \check{\bm{X}}^\top ( Z - \check{\bm{X}}  C^\diamond  )  - nm \lambda_{\lambda_K}  P^{-1}_{\ell} \Sigma  C^\diamond, \text{ and } \ 	C^\diamond = H_\lambda ( C^\diamond  + d^\diamond ),
\end{eqnarray}
	where $H_\lambda ( \cdot )$ is the hard thresholding operator that
	\begin{align}
		( H_\lambda (C))_ \ell = 
		\left \{ 
		\begin{array}{c}
			0, \quad if~  C_\ell^\top (P_\ell /nm ) C_\ell < 2 \lambda,  \\
			C_\ell, \quad if~ C_\ell^\top (P_\ell /nm ) C_\ell  \geq 2 \lambda,
		\end{array}
		\right.
	\end{align}
	where $P_\ell = \check X_\ell^\top \check X_\ell + nm \lambda_K \Sigma $.
	Conversely, if $ C^\diamond$ and $d^\diamond$ satisfy \eqref{eq:minimizer}, then $C^\diamond$ is a local minimizer of \eqref{eq:functional response loss}.
\end{lemma}

\begin{proof}
	Suppose $C^\diamond$ is a coordinate-wise minimizer of $L_\lambda$ in \eqref{eq:functional response loss}. 
Denote $\check X_\ell$ to be the the $\ell$-th column of $\check{\bm{X}}$. Then 
	\begin{eqnarray*}
		C^\diamond_\ell & \in & \arg \min_{t \in \mathbb{R}^m } {1 \over 2 n m } \|Z -  \check{\bm{X}} C^\diamond - \check X_l(t + C^\diamond_l) \|^2 + {\lambda_K \over 2 } t^\top \Sigma t + \lambda \|  t^\top K(u) \|_0  \\
		\Rightarrow  
		C^\diamond_\ell & \in & \arg \min_{t \in \mathbb{R}^m } {1 \over 2 n m }  \Big[ 
		t^\top ( \check X_\ell^\top \check X_\ell  + n m \lambda_K \Sigma ) t - 2 t^\top \check X_\ell^\top  ( \check X_\ell  C^\diamond  _\ell +  Z - \check X C^\diamond   )
		\Big] + \lambda \|  t^\top K(u) \|_0 \\
		\Rightarrow  
		C^\diamond_\ell & \in & {1 \over 2 } \Big[  t -  P_\ell ^{-1}  \check X_\ell^\top  ( \check X_\ell  C^\diamond  _\ell +  Z - \check X C^\diamond   ) \Big]^\top {P_\ell \over nm}
		\Big[  t -  P_\ell ^{-1}  \check X_\ell^\top  ( \check X_\ell  C^\diamond  _\ell +  Z - \check X C^\diamond   ) \Big] + \lambda \|  C_\ell^\top K(t) \|_0.
	\end{eqnarray*}
	where $P_\ell =( \check X_\ell^\top \check X_\ell  + n m \lambda_K \Sigma )  $. By the definition of the hard thresholding operator, we have
	$$
	C^\diamond_\ell  =  P_\ell ^{-1}  \check X_\ell^\top  ( \check X_\ell  C^\diamond  _\ell +  Z - \check X C^\diamond   )= C^\diamond  _\ell - nm \lambda_K  P_\ell ^{-1}   \Sigma  C^\diamond  _\ell 
	+ P_\ell ^{-1}  \check X_\ell^\top (Z - \check X C^\diamond   ) = H_\lambda( C^\diamond_\ell  + d^\diamond_\ell ),
	$$
	which shows \eqref{eq:minimizer} holds.
	
	Conversely, suppose that \eqref{eq:minimizer} holds. Let 
	$$
	A^\diamond = \{\ell \in S:  ( C^\diamond_\ell  + d_\ell  )^\top (P_\lambda  / n m ) ( C^\diamond_\ell  + d_\ell  ) \geq 2 \lambda  \}.
	$$
	By \eqref{eq:minimizer} and the definition of $H_\lambda (\cdot )$, we deduce that for $\ell \in A^\diamond, {C^\diamond_\ell }^\top (P_\lambda  / n m )  C^\diamond_\ell  \geq 2 \lambda $. Furthermore, $d^\diamond_{A^\diamond}=0$, which is equivalent to 
	$$
	C^\diamond_{A^\diamond} \in \argmin {1 \over 2 nm} 
	\sum_{i=1}^n \sum_{j=1}^m (  \bm{Z}_i (t_{j} ) -  \sum_{l \in {A^\diamond}} X_{i \ell }  K(t_{ij})^\top  C_\ell )^2
	+  {\lambda_K \over 2 }  \sum_{l \in {A^\diamond} } C_\ell^\top \Sigma C_l. 
	$$
	Next we show that $L_\lambda ( C^\diamond + h ) \geq L_\lambda ( C^\diamond) $ if $h$ is small enough .
	We consider two cases. If $h_{(A^\diamond)^c} \neq 0$, notice that $(C_\ell+h)^\top \Sigma (C_\ell + h) > C_\ell^\top \Sigma C_l$, then
	\begin{eqnarray*}
		L_\lambda ( C^\diamond + h ) -  L_\lambda ( C^\diamond)  &\geq& {1 \over 2 nm } \|  \bm{Z} - \check{\bm{X}} C^\diamond -  \check{\bm{X}}  h \|_2^2 -  {1 \over 2 nm } \|  \bm{Z} - \check{\bm{X}}  C^\diamond  \|_2^2
		+ \lambda + {\lambda_K  \over 2} \sum_\ell C_\ell \Sigma h  \\
		& \geq &
		\lambda - |  h^\top d^\diamond  |
	\end{eqnarray*}
	is positive for sufficiently small $h$. If $h_{(A^\diamond)^c} = 0$, by the minimizing property of $C^\diamond_{A^\diamond} $, we deduce that 
	$L_\lambda ( C^\diamond + h ) \geq L_\lambda ( C^\diamond) $. This complete the proof.
\end{proof}

\begin{lemma}
	\label{lem:iteration error}
	Suppose the conditions in Lemma \ref{eq:design_matrix_norm} holds,
	\begin{eqnarray}
	\label{eq: triangle_k}
\| \bigtriangleup^k \|_2  & \leq  & { \theta_{J,J}  tr(K) \over c_-(J) tr(K + \lambda_K I) } \| \bm{C}^*_{A_2^k}  \|_2 + 
h(J, \lambda_K ), \\
\label{eq: beta_iteration}
\|  \bm{C}^{k+1} - \bm{C}^*  \|_2 & \leq &
\Big( 1+ { \theta_{J,J}  tr(K) \over c_-(J) tr(K + \lambda_K I) } \Big)  D_2(A^k)+ 
h(J, \lambda_K ) ,
	\end{eqnarray}
 where $h(J, \lambda_K ) =\max_{|A|= J}  \| (T_{nm}^{A} + \lambda_K I)^{-1} ( (nm)^{-1}\sum_i \sum_j E_i(u_j)  K_{u_{ij}} X_{i, A} + \lambda_K \bm{C}^*_{A} )\|_2 $.
	\end{lemma}
\begin{proof}
Denote $T_{nm}^{A_R^k}(f)= (nm)^{-1} \sum_i \sum_j \langle X_{i, A_R^k}^\top f, K_{u_{ij}} \rangle_{\mathcal{K}}    K_{u_{ij}} X_{i, A_R^k}^\top $ and 
$ h_{nm}^{A_R^k} = (nm)^{-1} \sum_i \sum_j Z_{ij}   K_{u_{ij}} X_{i, A_R^k}$, where 
$X_{i, A_R^k}=( X_{i, \ell}, \ell \in A_R^k )^\top$.

Given the active set \( A_R^k \) at the \(k\)th iteration, the estimate at iteration \(k+1\) is updated by minimizing the following loss function, leveraging the reproducing property
$\langle C_{A_R^k}, K_{u_{ij}} \rangle_{\mathcal{K}} = C_{A_R^k}(u_{ij}),$
$$
{ 1 \over nm }\sum_i\sum_j  ( Z_{ij} - X_{i, A_R^k}^\top  \langle C_{A_R^k}, K_{u_{ij}} \rangle_{\mathcal{K}}   )^2 + {\lambda_K \over 2 } \sum_{\ell \in A_R^k} \| C_\ell \|^2_\mathcal{K},
$$
By taking the Fréchet derivative of the loss function with respect to \( C_{A_R^k} \), we obtain the estimate at iteration \(k+1\),
\begin{eqnarray*}
\bm{C}^{k+1}_{A_R^k} 
&=& 
(T_{nm}^{A_R^k} + \lambda_K I)^{-1} h_{nm}^{A_R^k}
=
(T_{nm}^{A_R^k} + \lambda_K I)^{-1}  (nm)^{-1} \sum_i \sum_j ( X_{i, A_1^k}\bm{C}^*_{A_1^k} +  X_{i, A_2^k}\bm{C}^*_{A_2^k}  + E_i(u_j) )  K_{u_{ij}} X_{i, A_R^k}
  \\
\bm{C}^*_{A_R^k} &=& (T_{nm}^{A_R^k} + \lambda_K I)^{-1} (T_{nm}^{A_R^k} + \lambda_K I)\bm{C}^*_{A_R^k} 
=
(T_{nm}^{A_R^k} + \lambda_K I)^{-1} (T_{nm}^{A_R^k} + \lambda_K I)\bm{C}^*_{A_1^k} \\
&=&
(T_{nm}^{A_R^k} + \lambda_K I)^{-1}  (nm)^{-1} \sum_i \sum_j \langle X_{i, A_1^k}^\top \bm{C}^*_{A_1^k}, K_{u_{ij}} \rangle_{\mathcal{K}}    K_{u_{ij}} X_{i, A_1^k}^\top
+
 \lambda_K (T_{nm}^{A_R^k} + \lambda_K I)^{-1} \bm{C}^*_{A_1^k}.
\end{eqnarray*}
According to Lemma \ref{lem:eigen value}, and the fact that $\|Tf\|_2 \leq \|T\|_{op} \|f\|_2$ for a linear operator $T$ and $f \in L_2 (\mathcal{T})$, it follows
{\footnotesize
\begin{eqnarray*}
\| \bigtriangleup^k \|_2
&=&
\| \bm{C}^{k+1}_{A_R^k}  - \bm{C}^*_{A_R^k}  \|_2
= 
\left \| (T_{nm}^{A_R^k} + \lambda_K I)^{-1}
\Big(
(nm)^{-1} \sum_i \sum_j ( X_{i, A_2^k}\bm{C}^*_{A_2^k}  + E_i(u_j) )  K_{u_{ij}} X_{i, A_R^k} 
+ \lambda_K \bm{C}^*_{A_1^k}
 \Big)
 \right \|_2 \\
 & \leq &
 { \theta_{J,J}  tr(K) \over c_-(J) tr(K + \lambda_K I) } \| \bm{C}^*_{A_2^k}  \|_2 + 
{ h(J,\lambda_K)  },
\end{eqnarray*} }
where $h(J, \lambda_K ) =\max_{|A|= J}  \| (T_{nm}^{A} + \lambda_K I)^{-1} ( (nm)^{-1}\sum_i \sum_j E_i(u_j)  K_{u_{ij}} X_{i, A} + \lambda_K \bm{C}^*_{A} )\|_2 $.

For the second inequality, according to the definition of $A_2^k$ and by the triangle inequality, we have
\begin{eqnarray*}
\|  \bm{C}^{k+1} - \bm{C}^*  \|_2
& \leq &
\|  \bm{C}^{k+1}_{A_R^k} - \bm{C}^*_{A_R^k}  \|_2 + \|  \bm{C}^*_{A^o \backslash A_R^k}  \|_2
=
\| \bigtriangleup^k \|_2 + \|  \bm{C}^*_{A_2^k}  \|_2 \\
&=&
\Big( 1+ { \theta_{J,J}  tr(K) \over c_-(J) tr(K + \lambda_K I) } \Big)  D_2(A_R^k)+ 
h(J,\lambda_K)   .
\end{eqnarray*}
This completes the proof.
\end{proof}

\begin{lemma}
If the conditions in Lemma \ref{eq:design_matrix_norm} holds, then
\begin{eqnarray}
\label{eq: D2 decomposition}
D_2 ( A_R^{k+1}  ) 
& \leq& 
\|   \bigtriangleup^k_{A_{11}^k  }   \|_2  + \|   \bm{C}^{k+1}_{A_{11}^k}   \|_2 \\
&&+
{1 \over c_-(J) } \Big(
 ( 2 \theta_{J,J} + \lambda) \|    \bigtriangleup^k_{A_R^k  } \|_2 +  \theta_{J,J} D_2 ( A_R^{k}  ) \Big) + {m \over  c_-(J)  tr(K+\lambda I)} (\|  d_{A_{22}^{k}}^k \|_2 + h(J)),   \nonumber  \\
\label{eq: d_I_44}
\|  d_{I_{44}^{k} }^{k+1}  \|_2
&\leq &
\big( m^{-1} tr(K) \theta_{J,J}  + \lambda \big)  \| \Delta_{A_R^k}^k  \|_2  + 
m^{-1} tr(K)  \theta_{J,J} D_2 (A_R^k)  +
h(J) + \lambda,  \\
\label{eq:D2 iteration}
D_2 ( A_R^{k+1}  )  &\leq &\gamma D_2(A_R^k) + \gamma m  h(J)  / \theta_{J,J} tr(K+\lambda I)  + \lambda.
\end{eqnarray}
\end{lemma}
\begin{proof}
	According to Lemma 24 of \cite{huang2018constructive}, we have 
\begin{eqnarray}
\label{eq: D2 ineuqalities}
D_2 ( A_R^{k+1}  )  \leq  \| \bigtriangleup^k_{A_R^k  } \|_2 + \|  \bm{C}^{k+1}_{A_{11}^k}   \|_2+ \|   \bm{C}^*_{A_{22}^k}   \|_2  . 
\end{eqnarray}
Recall that 
$d_{A_{22}^{k}}^{k+1}= -(nm)^{-1} (X_{A_{22}^{k}}^\top \otimes K_{\bm{u}}^\top  ) [\bm{Z} - X_{A_R^{k} } \otimes  K_{\bm{u}} \bm{C}_{A_R^k}^{k+1} ] + \lambda  \bm{C}_{ A_{22}^{k+1} }$
and $ \Delta_{A_R^k}^k  = \bm{C}_{A_R^k}^{k+1} - \bm{C}_{A_R^k}^* $,
one can derive that 
\begin{eqnarray*}
&& \|d_{A_{22}^{k}}^{k+1}\|_2 \\
&=&
\|  (nm)^{-1} (X_{A_{22}^{k}}^\top \otimes K_{\bm{u}}^\top  ) [  X_{A_R^{k} } \otimes  K_{\bm{u}} \Delta_{A_R^k}^k 
+ 
X_{A_R^{k} } \otimes  K_{\bm{u}} \bm{C}_{A_R^k}^*  
- X_{A^o } \otimes  K_{\bm{u}} \bm{C}_{A}^o - \bm{E} ] + \lambda_K  \bm{C}_{ A_{22}^{k+1} }  \|_2 \\
&=&
\|  (nm)^{-1} (X_{A_{22}^{k}}^\top \otimes K_{\bm{u}}^\top  ) [  X_{A_R^{k} } \otimes  K_{\bm{u}} \Delta_{A_R^k}^k 
-
X_{A_{2}^{k} } \otimes  K_{\bm{u}} \bm{C}_{A_{2}^k}^*  
 - \bm{E} ] + \lambda_K  \bm{C}_{ A_{22}^{k+1} }  \|_2 \\
 & = &
 \|  (nm)^{-1} (X_{A_{22}^{k}}^\top \otimes K_{\bm{u}}^\top  ) [  X_{A_R^{k} } \otimes  K_{\bm{u}} \Delta_{A_R^k}^k 
 - 
 X_{A_{22}^{k} } \otimes  K_{\bm{u}} \bm{C}_{A_{22}^k}^*  
 -
 X_{A_2^k \backslash  A_{22}^{k} } \otimes  K_{\bm{u}} \bm{C}_{A_2^k \backslash  A_{22}^k}^*  
 - \bm{E} ] + \lambda_K  \bm{C}_{ A_{22}^{k+1} }  \|_2 \\
 & \geq &
 m^{-1 } tr(K+\lambda_K I) \Big( c_-( J )  \| \bm{C}_{A_{22}^k}^*   \|_2 -
 (\theta_{J,J} +\lambda) \|\Delta_{A_R^k}^k \|_2 - \theta_{J,J}  \| \Delta_{A_R^k}^k  \|_2 - \theta_{J,J} D_2(A_R^k)
 \Big) 
 - h(J),
\end{eqnarray*}
where the first equality follows from the definition of $d^{k+1}$ and $\bm{Z}$; the second equality follows from the definition of $\Delta^k$; the third equality follows from simple algebra, and the last inequality follows from triangle inequality and the monotonicity property of $c_-(\cdot)$, $\theta_{a,b}$ and the definition of $h(J)=\max_{ A \subset S: |A|=J }  (nm)^{-1} \| ( X_A^\top \otimes K_{\bm{u}}^\top  )  \bm{E}  \|_2$. Plugging the above results into \eqref{eq: D2 ineuqalities}, we  have that \eqref{eq: D2 decomposition} holds.

Similarly, it is easy to deduce that for \eqref{eq: d_I_44}, 
\begin{eqnarray*}
 \|d_{I_{44}^k}^{k+1}\|_2 
	&=&
		\|  (nm)^{-1} ( X_{ I_{44}^k}^\top \otimes K_{\bm{u}}^\top  ) [  X_{A_R^{k} } \otimes  K_{\bm{u}} \Delta_{A_R^k}^k 
	-
	X_{A_2^{k} } \otimes  K_{\bm{u}} \bm{C}_{A_2^k}^*  
	- \bm{E} ] + \lambda_K  (\bm{C}_{ I_{44}^k} -  \bm{C}_{ I_{44}^k}^* ) + \lambda_K \bm{C}_{ I_{44}^k}^*  \|_2  \\
	& \leq & 
	m^{-1} tr(K+\lambda_K I) \big( \theta_{J,J} \| \Delta_{A_R^k}^k  \|_2  +   \theta_{J,J} D_2 (A_R^k)   \big) +h(J).
\end{eqnarray*}

Similar to Lemma 25 and Lemma 26 of \cite{huang2018constructive}, we  have 
$\|   \bm{C}^{k+1}_{A_{11}^k}   \|_2 +  \|  d_{A_{22}^{k+1} } \|_2 \leq \sqrt{2} ( \|    \bm{C}^{k+1}_{I_{33}^k}  \|_2 +  \|  d_{I_{44}^{k+1} } \|_2)$ and
$ \|   \bm{C}^{k+1}_{I_{33}^k}  \|_2    \leq \| \bigtriangleup^k_{I_{33}^k } \|_2$. It follows \eqref{eq: d_I_44} that
\begin{eqnarray*}
  &&  \| \bigtriangleup^k_{I_{33}^k } \|_2 +  \|  d_{I_{44}^{k+1} } \|_2
    \leq
    \|\Delta^k \|_2 + \|  d_{I_{44}^{k+1} } \|_2   \\
    &\leq& 
    (m^{-1} tr(K+\lambda_K I)  \theta_{J,J}   +1) \|\Delta^k \|_2 + m^{-1} tr(K+\lambda_K I)  \theta_{J,J}  D_2(A_R^k) + h(J).
\end{eqnarray*}
It then follows that 
\begin{eqnarray*}
D_2 ( A_R^{k+1}  ) 
& \leq& 
({ 2 \theta_{J,J} + \lambda \over c_-(J) } +1)
  \|    \bigtriangleup^k_{A_R^k  } \|_2 +  {\theta_{J,J} \over  c_-(J) } D_2 ( A_R^{k}  ) \\
  &&+ {m \over  c_-(J)  tr(K+\lambda I)}  h(J) +\max \{1,{m \over  c_-(J)  tr(K+\lambda I)} \}  
  ( \|   \bm{C}^{k+1}_{A_{11}^k}   \|_2 + \|d_{A_{22}^k}^k\|_2) \\
  &\leq &
 \Big( \max \{\sqrt{2},{m\sqrt{2} \over  c_-(J)  tr(K+\lambda I)} \}   (m^{-1} tr(K+\lambda_K I)  \theta_{J,J}   +1) + ({ 2 \theta_{J,J} + \lambda \over c_-(J) } +1) \Big) \|    \bigtriangleup^k_{A_R^k  } \|_2 \\
 &&+
 \Big( \max \{\sqrt{2},{m\sqrt{2} \over  c_-(J)  tr(K+\lambda I)} \}   m^{-1} tr(K+\lambda_K I) \theta_{J,J}  + {\theta_{J,J} \over  c_-(J) } \Big) D_2 ( A_R^{k}  ) \\
 &&+
\Big( \max \{\sqrt{2},{m\sqrt{2} \over  c_-(J)  tr(K+\lambda I)} \} +  {m \over  c_-(J)  tr(K+\lambda I)}   \Big) h(J)
\end{eqnarray*}
According to \eqref{eq: triangle_k} that
$\| \bigtriangleup^k_{A_R^k  } \|_2 \leq \Big( 1+ { \theta_{J,J}  tr(K) \over c_-(J) tr(K + \lambda_K I) } \Big)  D_2(A^k)+ 
h(J, \lambda_K ) $, then
$$
D_2 ( A_R^{k+1}  ) \leq \gamma D_2 ( A_R^{k}  )  + \gamma h(J, \lambda_K),
$$
where
\begin{eqnarray}
\label{eq:gamma_C_def}
\gamma&=& \Big( \max \{\sqrt{2},{m\sqrt{2} \over  c_-(J)  tr(K+\lambda I)} \}   (m^{-1} tr(K+\lambda_K I)  \theta_{J,J}   +1) + ({ 2 \theta_{J,J} + \lambda \over c_-(J) } +1) \Big)  \nonumber \\
&&\cdot \Big( 1+ { \theta_{J,J}  tr(K) \over c_-(J) tr(K + \lambda_K I) } \Big).
\end{eqnarray}

\end{proof}

\begin{lemma}
	\label{lem: Tn_T}
	If the conditions in Assumptions \ref{assum: J}-\ref{assump:error of imaging} are satisfied, and data splitting is conducted such that the outcome model and the exposure model are estimated using different samples.
    For $T=K^{1/2} E( \tilde{\bm{Z}}(t) \tilde{\bm{Z}}(s)  ) K^{1/2}$ and $T_n=K^{1/2}  C_n(s,t) K^{1/2}$ with $C_n(s,t)=n^{-1/2} \sum_{i=1}^n \widehat{\bm{Z}}_i(s) \widehat{\bm{Z}}_i(t) $, then for
    any $\nu \in (0, 1)$
	\ba
	\label{eq: T_half_Tn_T_op}
	P \Big( \|  ( T + \lambda I  )^{-1/2} ( T_n  - T)     \|_{op}  \leq  ( 1 + {1 \over \sqrt{ \nu}} ) \sqrt{ {c_1 \over n} tr(T) tr( T(T+\lambda )^{-1} ) } \Big) 
	\geq 
	1-\nu, \\
	\label{eq: Tn_T_op}
	P \Big( \|  ( T_n  - T)     \|_{op}  \leq  ( 1 + {1 \over \sqrt{ \nu}} ) \sqrt{ {c_1 \over n} tr^2(T) } \Big) 
	\geq 
	1-\nu.
	\ea
\end{lemma}	
\begin{proof}
	Recall that by Mercer's Theorem, the operator kernel $T$ admits the spectra decomposition $ T (s, t) = \sum_{k=1}^\infty s_k \varphi_k(s)  \varphi_k(t)$,
	where $s_1 > s_2 > \cdots$ are eigenvalues of $T$, and $\{ \varphi_k \}$ are the eigenfunctions.
	Then for $h \in L_2 (\mathcal{T}), \| h \|_{L_2} =1$, one can have $ h (t)= \sum_{k=1}^\infty h_k \varphi_k (t)$.
	
	By definition, we can see that 
	\bas
	&& \|  ( T + \lambda I  )^{-1/2} ( T_n  - T)     \|_{op}
	=
	\sup_{h: \| h \|_{L_2} =1} | \langle  ( T + \lambda I  )^{-1/2} h, ( T_n  - T)h \rangle    |  \\
	& =&
	\sup_{h: \| h \|_{L_2} =1}  \sum_{j,k} { h_j h_k \over ( s_j + \lambda )^{1/2} } \langle \varphi_j, (T_n- T) \varphi_k   \rangle  \\
	& \leq &
	\sup_{h: \| h \|_{L_2} =1}  ( \sum_{j,k} h_j^2 h_k^2 )^{1/2} 
	\big(    \sum_{j,k}   { 1\over  s_j + \lambda }  \langle \varphi_j, (T_n- T) \varphi_k   \rangle ^2    \big)^{1/2} \\
	& \leq &
	\big(    \sum_{j,k} { 1\over  s_j + \lambda  }  \langle \varphi_j, (T_n- T) \varphi_k   \rangle ^2    \big)^{1/2} := H_n,
	\eas	
	where $H_n$ corresponds to the last term.
	
	Denote $\tilde C_n(s,t)=n^{-1/2} \sum_{i=1}^n \tilde{\bm{Z}}_i(s) \tilde{\bm{Z}}_i(t)$ and $\tilde T_n= K^{1/2} \tilde C_n K^{1/2}$,Take the expectation of $H_n$, then
	\bas
	E( H_n ) &\leq&  \big(    \sum_{j,k} { 1\over  s_j + \lambda  } E  \langle \varphi_j, (T_n- T) \varphi_k   \rangle ^2    \big)^{1/2} \\
    &=&
    \big(    \sum_{j,k} { 1\over  s_j + \lambda  } E  \langle \varphi_j, ( T_n- \tilde T_n  \rangle ^2    \big)^{1/2}
    +\big(    \sum_{j,k} { 1\over  s_j + \lambda  } E  \langle \varphi_j, (\tilde T_n - T) \varphi_k   \rangle ^2    \big)^{1/2} \\
    &=&
    H_{n1} + H_{n2},
	\eas
    where $H_{n1}$ and $H_{n2}$ denote the above two terms, respectively.
    
	Direct calculations lead to 
	\bas
	&& E  \langle \varphi_j, (T_n- \tilde T_n) \varphi_k   \rangle ^2
	=
	E \Big(
	\int_{\mathcal{T}} \int_{\mathcal{T}} K^{1/2} \varphi_j(s)( C_n(s, t) - \tilde C_n(s, t) ) K^{1/2} \varphi_k(t) ds dt \Big)^2 \\
	& =&
	{1\over n} E \Big(
	\int_{\mathcal{T}} \int_{\mathcal{T}} K^{1/2} \varphi_j(s) [ \widehat{\bm{Z}}(s) \widehat{ \bm{Z} }(t) -\widetilde{ \bm{Z} }(s) \widetilde{\bm{Z}}(t)  ]K^{1/2} \varphi_k(t) ds dt 
	\Big)^2   \\
    &=&
    {1\over n} E \Big(
	\int_{\mathcal{T}} \int_{\mathcal{T}} K^{1/2} \varphi_j(s) [ X_i^\top \widehat{\bm{C}}(t) \widehat{\bm{C}}^\top(s) X_i -X_i^\top {\bm{C}}(t) {\bm{C}}^\top(s) X_i  ]K^{1/2} \varphi_k(t) ds dt \\
    &=&
   {1\over n} E \Big(
	\int_{\mathcal{T}} \int_{\mathcal{T}} K^{1/2} \varphi_j(s) X_i^\top [  \big( \widehat{\bm{C}}(t) -\bm{C}(t) \big) \big(\widehat{\bm{C}}(s) -\bm{C}(s) \big)^\top +   2 \big(\widehat{\bm{C}}(t) -\bm{C}(t) \big) \bm{C}^\top(s)   ] X_i K^{1/2} \varphi_k(t) ds dt  \\ 
	& = &
	{1\over n} E \Big(
	\int_{\mathcal{T}} \int_{\mathcal{T}} K^{1/2} \varphi_j(s) \widetilde{\bm{Z}}(s) [  \big( \frac{\widehat{\bm{C}}(t) -\bm{C}(t)}{\bm{C}(t)} \big) \big(\frac{ \widehat{\bm{C}}(s) -\bm{C}(s)} {\bm{C}(s)} \big)^\top +   2 \big( \frac{ \widehat{\bm{C}}(t) -\bm{C}(t)}{\bm{C}(t)} \big)    ]
    \widetilde{\bm{Z}} (t)  K^{1/2} \varphi_k(t) ds dt 
	\Big)^2    \\
	& \leq &
	{1\over n} E ^{1/2} \Big(
	\int_{\mathcal{T}} \int_{\mathcal{T}} K^{1/2} \varphi_j(t) \widetilde{\bm{Z}}(t) dt \Big)^4 
	E ^{1/2} \Big(
	\int_{\mathcal{T}} \int_{\mathcal{T}} K^{1/2} \varphi_k(t) \widetilde{\bm{Z}}(t) dt \Big)^4     \\
	& \leq & 
	{c_1 \over n}
	E \Big(
	\int_{\mathcal{T}} \int_{\mathcal{T}} K^{1/2} \varphi_j(t) \widetilde{\bm{Z}}(t) dt \Big)^2 
	E \Big(
	\int_{\mathcal{T}} \int_{\mathcal{T}} K^{1/2} \varphi_k(t) \widetilde{\bm{Z}}(t) dt \Big)^2 ,
	\eas
	where the first inequality follows from that $\sup_t (\widehat{\bm{C}}(t) -\bm{C}(t))/\bm{C}(t) =o_p(1)$ due to $\|\widehat{\bm{C}}(t) -\bm{C}(t) \|_2^2=o_p(1)$, and
    the last inequality follows directly from Assumption \ref{assum: 4-th order}.
	It is obvious that $ E \Big(
	\int_{\mathcal{T}} \int_{\mathcal{T}} K^{1/2} \varphi_j(t) \widetilde{\bm{Z}} Z(t) dt \Big)^2 =\langle T \varphi_j, \varphi_j \rangle = s_j  $,
	then 
	\bas
	E  \langle \varphi_j, (T_n- \tilde T_n) \varphi_k   \rangle ^2 \leq  c_1 s_j s_k/n. 
	\eas
	Hence, we obtain that
	\bas
	E( H_{n1} ) 
	& \leq &   ({ c_1 \over n})^{1/2}  \Big( \sum_{j,k} {   s_j s_k \over  s_j + \lambda  }  \Big)^{1/2} =
	({ c_1 \over n})^{1/2} \Big(  tr(T) tr( T(T + \lambda I)^{-1} ) \Big)^{1/2} ,  \\
	E(  H_{n1}^2 )  &\leq &
	{ c_1 \over n}  tr(T) tr( T(T + \lambda I)^{-1} ) .
	\eas
	According to the concentration inequity that
	\bas
	P( |  H_{n1} - E  H_{n1}| \geq t ) \leq {E  H_{n1}^2 \over t^2},
	\eas
	taking $t = {1 \over \sqrt{ \nu}} \sqrt{ c_1 tr(T) tr( T(T + \lambda I)^{-1} ) /n }$, we have
	\bas
	P \Big( H_{n1}  \geq  ( 1 + {1 \over \sqrt{ \nu}} ) \sqrt{ {c_1 \over n} tr(T) tr( T(T+\lambda )^{-1} ) }  \Big)  
	\leq 
	\nu.
	\eas
    Similarly, we can show that 
    \bas
	P \Big( H_{n2}  \geq  ( 1 + {1 \over \sqrt{ \nu}} ) \sqrt{ {c_1 \over n} tr(T) tr( T(T+\lambda )^{-1} ) }  \Big)  
	\leq 
	\nu.
    \eas
	Then it follows with probability at least $1 - \nu$, 
	\bas
	\|  ( T + \lambda I  )^{-1/2} ( T_n  - T)     \|_{op} \leq H_n
	\leq
	( 1 + {1 \over \sqrt{ \delta}} ) \sqrt{ {c_1 \over n} tr(T) tr( T(T+\lambda )^{-1} ) }  .
	\eas
	Hence, \eqref{eq: T_half_Tn_T_op} holds. By similar steps, we can prove that \eqref{eq: Tn_T_op} holds and we omit here. 
\end{proof}

\bibliographystyle{chicago} 
\bibliography{RefFLSEM}